\documentclass[10pt]{article}
\usepackage[utf8]{inputenc}
\usepackage{custommacros}

\title{Streaming Graph Algorithms in \\
        the Massively Parallel Computation Model
    \thanks{{A preliminary version of this paper appeared in \emph{Proceedings of the 43rd Annual ACM Symposium on Principles of Distributed Computing (PODC'2024)}, pages 496--507, 2024.}}
}

\author{
\begin{tabular}[t]{c c c}
\textbf{Artur Czumaj}\thanks{Department of Computer Science and Centre for Discrete Mathematics and its Applications (DIMAP), University of Warwick, UK. E-mail: A.Czumaj@warwick.ac.uk. Research supported in part by the Centre for Discrete Mathematics and its Applications (DIMAP), by EPSRC award EP/V01305X/1, by a Weizmann-UK Making Connections Grant, and by an IBM Award.} &
\textbf{Gopinath Mishra}\thanks{Department of Computer Science, National University of Singapore, Singapore. E-mail:  Gopinath@nus.edu.sg. The work was partly done when the author was a Postdoc at the University of Warwick and the research was supported in part by the Centre for Discrete Mathematics and its Applications (DIMAP) and by EPSRC award EP/V01305X/1.}  &
\textbf{Anish Mukherjee}\thanks{Department of Computer Science, University of Liverpool, UK. E-mail: Anish.Mukherjee@liverpool.ac.uk. The work was partly done when the author was at the University of Warwick and the research was supported in part by the Centre for Discrete Mathematics and its Applications (DIMAP) and by EPSRC award EP/V01305X/1.}
\end{tabular}
}

\begin{document}
\maketitle
\begin{abstract}
We initiate the study of graph algorithms in the streaming setting on massive distributed and parallel systems inspired by practical data processing systems. The objective is to design algorithms that can efficiently process evolving graphs via large batches of edge insertions and deletions using as little memory as possible.

We focus on the nowadays canonical model for the study of theoretical algorithms for massive networks, the \emph{Massively Parallel Computation (MPC) model}. We design MPC algorithms that efficiently process evolving graphs: in a constant number of rounds they can handle large batches of edge updates for problems such as connectivity, minimum spanning forest, and approximate matching while adhering to the most restrictive memory regime, in which the \emph{local memory} per machine is \emph{strongly sublinear in the number of vertices} and the \emph{total memory} is \emph{sublinear in the graph size}. These results improve upon earlier works in this area which rely on using larger total space, proportional to the size of the processed graph. Our work demonstrates that parallel algorithms can process dynamically changing graphs with asymptotically optimal utilization of MPC resources: parallel time, local memory, and total memory, while processing large batches of edge updates.
\end{abstract}

\section{Introduction}
\label{intro}

In the last decades analyzing massive graphs and networks has become an important part of many areas of computing and its applications, including social network analysis, machine learning, medical applications, and others. The challenges of efficiently processing such graphs lie not only in their \emph{massive size}, reaching even billions and trillions of nodes or edges (see, e.g., \cite{CEKLM15,KH16,RBMZ15,SWL13,WYXXMWLDZ15}), but also in their \emph{dynamic character} (see, e.g., \cite{BFKKH23,HHS22}). For example, the dynamic nature of social networks and the steadily evolving structure of the Web often require to cope with networks where millions of edges may be added or removed per second;
similar scenarios 
appear naturally in the analysis of retail transactions, protein interaction, etc.

In order to efficiently process massive amounts of data, where the input does not fit into the available memory of even the most advanced modern systems, we naturally have to rely on distributed or parallel systems. Most prominently, parallel computation and storage systems, such as MapReduce \cite{DG10}, Hadoop \cite{White15}, Dryad \cite{IBYBF07
}, or Spark \cite{ZCFSS10}, 
have been successfully used to represent, process, and analyze the massive datasets that appear in many modern applications. As an elegant abstraction of these frameworks, the \emph{Massively Parallel Computation} (\MPC) model, first introduced by Karloff \etal \cite{KSV10} and later refined in \cite{ANOY14,BKS17,GSZ11}, has become the canonical theoretical model of the study of parallel algorithms. At a very high level, an \MPC system consists of a collection of machines that can communicate with each other through indirect communication channels. The computation proceeds in synchronous rounds, where at each round the machines receive messages from other machines, perform local computations, and finally send appropriate messages to other machines so that the next round can start. The central factor typically to be optimized in the analysis of efficient algorithms in the \MPC model is the \emph{number of rounds} while maintaining (as a hard constraint) the \emph{low local and global capacity} of the machines used, and \emph{low communication} performed per round.

While the study of \MPC algorithms has led to major advances in parallel algorithms in the static setting
, only a limited amount of research has been done in the context of dynamically changing systems, where even the execution of very efficient algorithms after a few updates in the input data might be prohibitive due to their large processing time and resource requirements. This is in contrast to the traditional sequential model of computation, where various algorithms have been successfully exploiting the fact that small modifications in the input data often have a very small impact on the solution, compared to the solution in the input instance prior to the modifications.

We consider a natural model adapting the \MPC model in a dynamic environment as a distributed system, where
%
\begin{itemize}
\item each machine stores some part of the input or its representation (e.g., a sketch), and
\item in each synchronous phase the system maintains a solution while allowing multiple modifications to the input, such as insertions or deletions of elements in the maintained dataset.
\end{itemize}
%
We assume that the updates of the data arrive in a distributed fashion, any machine can receive a request for some modifications (a similar case can be made for coordinated updates). Following related research on the \PRAM model of parallel computation (see, e.g., \cite{AABD19,GL20,TDB19}), such a model has been introduced recently in the context of graph problems on \MPC by Italiano \etal \cite{ILMP19}, who considered only a single update per phase. Dhulipala \etal \cite{DDKPSS20} and Nowicki and Onak \cite{NO21} extended this setting to a more natural case of multiple updates. (Indeed, as it was motivated and argued well by Dhulipala \etal \cite{DDKPSS20}, in typical applications of dynamic massive datasets one would expect numerous updates of the datasets with high frequency.) 

Similarly to the sequential model, the goal of an efficient algorithm is to maintain a solution to a problem more efficiently than recomputing the solution from scratch. 
Here, in the \MPC setting, the main goal is to reduce the bounds in key factors contributing to the complexity of an algorithm.
\begin{itemize}
\item \textbf{Local memory} used by individual \MPC machines.
\item \textbf{Total memory} used in \MPC system, which is the sum of the storage available on all \MPC machines (equivalently, a product of the number of machines and the local memory).
\item \textbf{Number of updates} (also called a \emph{batch size}) allowed in a single phase.
\item \textbf{Update time}, which is the \emph{number of rounds} performed between updates, in a single phase.
\item \textbf{Query time}, which is the \emph{number of rounds} to compute a solution to the problem instance at the end of a phase.
\end{itemize}

Since our focus is on graph algorithms, for simplicity, we assume that the algorithm starts with an empty graph $G$ (possibly weighted) with no edges (but with a fixed set of $n$ vertices $V$), and then, it dynamically evolves by edge insertions and deletions. (We notice that one can relax this assumption and all our algorithms can start with an arbitrary graph and preprocess efficiently, see the end of \Cref{subsec:contributions}).
We will denote the \emph{current number of edges} of $G$ by $m$ ($m$ may vary in time).

As it is nowadays standard for \MPC algorithms, since local memory is related to a natural 
hardware limitation, we will be assuming that the \emph{local memory} \lspace is \emph{sublinear} in the input size, and in order to be scalable%
\footnote{This model is often called a \emph{fully-scalable \MPC} to emphasize that it can easily scale up with the size of the input: if the existing hardware offers local memory $M$ then the sublinear setup is well parameterized to study by choosing~an~appropriate~$\spacexp$.}%
, we will consider \lspace to be upper bounded by $n^{\spacexp}$ for some arbitrarily small positive constant \spacexp (optimally, arbitrarily close to 0). In regards to the number of updates (or a batch size), while \cite{ILMP19} considered only a single update and \cite{DDKPSS20} allowed at most $n^{\spacexp-\Theta(1)}$ updates, we will study a more general case (as in \cite{NO21}), in which the number of updates allowed is $\tOh(n^{\spacexp}) = \Oh(n^{\spacexp}/\polylog(n))$, that is, the number of updates may be almost as large as the capacity of a single \MPC machine\footnote{We also notice that the arguments from Section 2 in \cite{NO21} demonstrate that with the target of constant-rounds algorithms, it is unlikely to significantly increase the number of updates per phase.}. The \emph{update time} and the \emph{query time} are the main parameters to be optimized by a parallel algorithm. We aim to design algorithms that run in a \emph{constant number of rounds} (when \spacexp is a constant) in a phase (between updates), for both, update time and query time.

The \emph{total memory} is the \emph{key parameter that distinguishes our work from earlier papers} \cite{ILMP19,DDKPSS20,NO21}. Earlier algorithms were using $\widetilde{\Theta}(n+m)$ total memory across all \MPC machines, allowing to store the entire graph and additional data structures in the system; this corresponds to the classical setting of \emph{dynamic algorithms} \cite{HHS22} (with the dynamically changing number of edges, one uses extra machines for newly arriving edges or bounds $m$).
In our work, we follow the classical approach of \emph{streaming algorithms} \cite{CY20,FKMSZ05,McGregor14,Muthukrishnan05}\footnote{The focus in this paper is on modeling \emph{single-pass streaming}, though one could also consider multiple-pass streams.}, and instead of focusing solely on minimizing time to conduct graph updates, \emph{our primary goal is to combine low (constant) update and query times with \textbf{sublinear total memory} needed to store the data}.

There is a close link between the space used in graph streaming algorithms and \MPC algorithms in the setting discussed here, in that any streaming space lower bound for graph algorithms immediately yields the same total memory lower bound in our setting\footnote{Observe that a dynamic \MPC algorithm that uses total memory $M$ can be trivially implemented as a streaming algorithm that uses the same space $M$, just by simulating sequentially the \MPC algorithm.}. Therefore, since most of the work on graph streaming has been focusing on the \emph{semi-streaming model} \cite{FKMSZ05, McGregor14}, we will extend this model to the \MPC setting. In the semi-streaming model, the data stream algorithm is permitted $\tOh(n) = \Oh(n \polylog(n))$ memory, where $n$ is the number of vertices (and thus the memory is independent of the number of edges $m$). The focus on semi-streaming algorithms is on one hand because most problems are \emph{provably intractable if the available memory is sublinear in $n$}, and at the same time, many problems become \emph{feasible once there is memory roughly proportional to the number of vertices} in the graph. Following this line of research, we consider \MPC algorithms for evolving graphs where the \MPC system is permitted 
$\Oh(n \polylog(n))$ total memory.

We can now state the main question studied in this paper:

\medskip
\centerline{%
\parbox{5.5in}{
\begin{mdframed}[hidealllines=true,backgroundcolor=gray!15]
    \emph{Can we design MPC algorithms for fundamental problems in dynamically evolving graphs that with sublinear local memory $n^{\spacexp}$ and total memory $\Oh(n \polylog(n))$, can maintain in a constant number of rounds good solutions to the problem after $\Oh(n^{\spacexp})$ updates?}
\end{mdframed}
}%
}

\medskip

We believe that the question above is very natural, addresses important features of modern parallel systems, and is also scientifically challenging. We have been arguing above that massive graphs should be studied in the parallel setting, and the \MPC model with sublinear local memory seems to be among the most suitable frameworks for such a study. Further, we claim that the focus on constant-round updates is a natural requirement from very powerful parallel systems modeled by \MPC. So, for example, while in the classical streaming setting one is usually less concerned about the time needed to process a single update, in the parallel system, it might be acceptable to have slow algorithms on individual machines, but because of its significant cost, we want to optimize the number of communication rounds to process the updates. Finally, we argue that the total memory, or equivalently, the number of machines, is a very expensive resource that ought to be minimized. This is especially important in the era of massive graphs, where we want to be ready to manage graphs with trillions of edges, or more, and committing so big resources for the task is often unlikely or overly expensive. (Moreover, our requirement that the total memory is $\tOh(n)$ ensures the full independence on the number of edges in the evolving graph, which makes it easier to manage available resources since in general, the number of edges may vary in time.)

\subsection{Our contributions}
\label{subsec:contributions}

The main finding of this paper is an affirmative answer to the question above for three fundamental graph problems: connectivity, minimum spanning forest, and maximum matching.

We begin with our 
main result for maintaining connectivity.

\begin{theorem}
\label{thm:intro-con}
Let $0 < \spacexp < 1$ be an arbitrary constant. Given an undirected graph $G$ with $n$ vertices, we can maintain the \textbf{connectivity} of $G$  to process a batch of $\tOh(n^{\spacexp})$ updates in
a constant number of
rounds on an \MPC with sublinear local memory $\Oh(n^{\spacexp})$ and {$\tOh(n)$ total memory}.

Furthermore, within the same bounds, the algorithm can maintain a \textbf{spanning forest} of $G$.
%
%
\end{theorem}

Extending the ideas from \Cref{thm:intro-con}, we can process updates for dynamic approximate minimum spanning forest and an exact minimum spanning forest algorithm in insertion-only streams.



\begin{theorem}
\label{thm:intro-mst}
Let $0 < \spacexp < 1$ be an arbitrary constant. Given a weighted graph $G$ with $n$ vertices, on an \MPC with $\lspace=\Oh(n^{\spacexp})$ local memory and $\widetilde{\Oh}(n)$ total memory we can process a batch of $\tOh(n^{\spacexp})$ updates and maintain in a constant number of rounds

\begin{itemize}
    \item an exact {minimum spanning forest} for insertion-only updates, and
    \item $(1+\epsilon)$-approximate  {minimum spanning forest} for arbitrary updates.
\end{itemize}

%
\end{theorem}

Since there is no semi-streaming algorithm for a constant approximation of matching for general dynamic streams (see, e.g., \cite{AKLY16}), 
our results for matching are slightly weaker.

\begin{theorem}
\label{thm:intro-match}
Let $0 < \spacexp < 1$ be an arbitrary constant and let $0 < \kappa < \spacexp$ and $\alpha > 1$ be arbitrary. Given an undirected graph $G$ with $n$ vertices, on an \MPC with $\lspace = \Oh(n^{\spacexp})$ local memory we can process a batch of $\tOh(\lspace^{1 - \kappa})$ updates and maintain in $\Oh(\log(1/\kappa))$ rounds an
\textbf{$\Oh(\alpha)$-approximate maximum matching} in $G$ with
\begin{itemize}
\item $\widetilde{\Oh}(n/\alpha)$ total memory for insertion-only updates, and
\item $\widetilde{\Oh}(\max\{n^2/\alpha^3,n/\alpha\})$ total memory for arbitrary updates.
\end{itemize}
Furthermore, if we are only required to output an estimate on the \textbf{size of the matching}, then the total memory bounds can be improved to $\widetilde{\Oh}(n/\alpha^2)$ and $\widetilde{\Oh}(n^2/\alpha^4)$, respectively, where $\alpha \leq \sqrt{n}$.
\junk{
Let $0 < \spacexp < 1$ be a constant and let $\alpha \ge 1$ be an arbitrary positive real. Given an undirected graph $G$ with $n$ vertices we can find an $\Oh(\alpha)$ approximate maximum matching in $G$ in single pass \smpc with $\lspace = \Oh(n^\delta)$ local memory in $\Oh(\log(1/\kappa))$ rounds to process a batch of updates of size $\Oh(\lspace^{1-\kappa})$ with
\begin{itemize}
\item $\widetilde{\Oh}(n/\alpha)$ global memory for insertion-only streams, and
\item $\widetilde{\Oh}(n^2/\alpha^3)$ global memory for dynamic streams.
\end{itemize}
For constant $\kappa$ this gives an $\Oh(1)$ round algorithm and for $\kappa = \Oh(1/\log n)$ an $\Oh(\log \log n)$ round algorithm. In addition, if we are only required to output an estimate on the size of the matching, then the global memory bounds in the above can be improved to $\widetilde{\Oh}(n/\alpha^2)$ and $\widetilde{\Oh}(n^2/\alpha^4)$, respectively.
}
\end{theorem}

In particular, \Cref{thm:intro-match} implies the following results
.

\begin{corollary}
\label{cor:intro-match-low-total}
Let $0 < \eps < \spacexp < 1$ be arbitrary constants.
Given an undirected graph $G$ with $n$ vertices, on an \MPC with $\lspace = \Oh(n^{\spacexp})$ local memory and $\widetilde{\Oh}(n)$ total memory we can process a batch of $\tOh(n^{\spacexp - \eps})$ insertion-only updates and maintain in a constant number of rounds a \textbf{$\Oh(1)$-approximate maximum matching} in $G$.
\end{corollary}

\begin{corollary}
\label{cor:intro-match-higher-total}
Let $0 < \spacexp < 1$ be an arbitrary constant and $\alpha > 1$ be arbitrary.
Given an undirected graph $G$ with $n$ vertices, on an \MPC with $\lspace = \Oh(n^{\spacexp})$ local memory we can process a batch of $\tOh(n^{\spacexp})$ updates and maintain in $\Oh(\log\log n)$ rounds an \textbf{$\Oh(\alpha)$-approximate maximum matching} 
with
\begin{itemize}
\item $\widetilde{\Oh}(n/\alpha)$ total memory for insertion-only updates, and
\item $\widetilde{\Oh}(\max\{n^2/\alpha^3,n /\alpha\})$ total memory for arbitrary updates.
\end{itemize}
\end{corollary}

All our claims above are for randomized algorithms and hold as long as the total number of updates is polynomial in $n$.
The algorithms allow the (batch) update to the streaming to be \emph{adversarial}, though the adversary is \emph{oblivious} (e.g., they are not adversarially robust \cite{BJWY22}).

In all claims above, the amount of \emph{global communication} in a single \MPC round is upper bounded by total memory used, and thus it is $\tOh(n)$ in Theorems~\ref{thm:intro-con}--\ref{thm:intro-mst} and in Corollary~\ref{cor:intro-match-low-total}.

The \emph{query time} in all claims above (Theorems~\ref{thm:intro-con}--\ref{thm:intro-match} and Corollaries~\ref{cor:intro-match-low-total}--\ref{cor:intro-match-higher-total}) is constant. This follows from the fact that in each of the problems studied we maintain the solution to the problem at hand. And so, for connectivity, the algorithm maintains vertex labeling such that two vertices have the same label if and only if they are in the same connected component. Therefore reporting the connected components can be easily done by sorting the labels (which can be done in $\Oh(1)$ \MPC rounds, see \cite{GSZ11}). Similarly, for a spanning forest, a minimum spanning forest (or its weight approximation), and an approximate maximum matching, we can maintain the list of all edges in the solution, which can be then reported when needed in a constant number of rounds.

Observe that all our algorithms above have their \emph{total memory asymptotically optimal up to polylogarithmic factors} since it matches the state-of-the-art space bounds of the best streaming algorithms (or lower bounds) for the corresponding problems, see \Cref{subsec:related} for more details.

For simplicity, we assume that the algorithms start with an empty graph on $n$ vertices $V$ (which is a standard setting in dynamic and streaming algorithms). However, our algorithms can be easily modified to start the computation at any point of time and in the ``pre-computation phase'' to solve the instance (of size $\poly(n)$) using known static algorithms. For connectivity and MST this can be done in $\Oh(\log n)$ rounds \cite{AGM12, NO21} and for matching in $\tOh(\sqrt{ \log n})$ rounds \cite{GU19,Onak18}. Hence, the main challenge is to process the updates considerably faster after a batch of insertions/deletions dynamically than by running the static algorithm from scratch every time.
%

\subsection{The model in details}
\label{subsec:model}

We follow the standard terminology of streaming algorithms for graphs \cite{BFKKH23,FKMSZ05,McGregor14} and of parallel algorithms for dynamic graphs \cite{AABD19,DLSY21,ILMP19,NO21,TDB19,TDS22}. We start with an empty graph $G$ with no edges but with a fixed vertex set $V$. Then, 
at each algorithm step a new edge is inserted into the graph or an existing edge is deleted. Since we consider multiple edge insertions and deletions (an update \emph{batch}), in a single phase the graph undergoes a batch of insertions and deletions, and in the initial round of each computation, an update or query batch is arriving at the system arbitrarily distributed among the machines. We will assume that at any moment, the current graph is simple and all deletions concern only the existing edges. However, minor modifications to our algorithms would also take care of the presence of parallel edges.

The underlying computational model is the \emph{Massively Parallel Computation} (\MPC) model (see, e.g., \cite{ANOY14,ASSWZ18,BKS17,CDP24,GSZ11,GU19,KSV10}), which corresponds to the available hardware in the system. {An \MPC system consists of a collection of machines that can communicate with each other through indirect communication channels.} In order to allow broad scalability of our algorithms, we assume that the size \lspace of the local memory of each machine is strongly \emph{sublinear} with respect to the number of vertices of the graph and is $\lspace = n^{\spacexp}$ for\footnote{To simplify the notation, we assume that \spacexp is a constant but our algorithms will work for any \spacexp arbitrarily close to 0. If \spacexp is a function on $n$, then our algorithms would typically run in $\Oh(1/\spacexp)$ rounds (which is $\Oh(\log_{\lspace}n)$). It is also known that on an \MPC with local memory $n^{\spacexp}$, essentially any problem needs $\Omega(1/\spacexp)$ rounds, see \cite{RVW18}.} an arbitrarily small positive constant \spacexp. We parameterize the \MPC model with respect to the total memory, and since we want to optimize this parameter, as in graph streaming problems, we will aim to allow only $n \polylog(n)$ total memory (and hence the number of machines is $\tOh(n/s) = \tOh(n^{1-\spacexp})$), though we relax this target in the analysis for matching, see \Cref{thm:intro-match}. Then, the computation on an \MPC takes place in synchronous rounds: in each round, each machine may perform arbitrary computation on its local memory, and then exchange messages with other machines. Each message is sent to a single machine specified by the machine sending the message. Furthermore, the total messages sent or received by each machine in each round should not exceed its memory \lspace (hence, in a round the total number of messages sent is $\tOh(n)$). The messages are delivered at the start of the next round. If at any moment the solution is requested, it is to be output collectively by machines (e.g., put all edges of a solution on the first $\tOh(n/s) = \Oh(n^{1-\spacexp})$ machines).

In order to use the \MPC model for a dynamically evolving graph, we will assume that the rounds are consolidated into phases, and at the beginning of each phase a batch of updates (insertions of new edges $\{u,v\}$ and deletions of existing edges $\{u,v\}$; the edges may be weighted) arrives to the system, and possibly, at the end of any phase the system must report a solution to the problem at hand. We will assume that the updates can arrive arbitrarily in the system, i.e., at the beginning of any round some machines receive update requests for some edges.
It is important to note that while the (batch) update to the streaming is adversarial, the adversary is oblivious. In order for the model to cope with update batches and to run efficiently, we will bound the number of edge updates per batch, which we parameterize, aiming to allow the batch size to be as large as {$\tOh(\lspace)$} (see Section 2 in \cite{NO21} for some arguments why large batches will likely require $\omega(1)$ rounds). Finally, if the updates have insertions only, we will call such instances \emph{insertion-only}.

We consider the setup in which the updates (at most {$\tOh(\lspace)$} of them) of edge insertions and edge deletions can be arbitrarily distributed among the \MPC machines. However, we can assume without loss of generality that all updates are arriving on a single \MPC machine. This follows from a known fact \cite{GSZ11} that sorting of $N$ objects can be performed in a constant number of rounds on an \MPC with sublinear local memory and total memory $\Oh(N)$. Using this fact, we can preprocess the data and move all update requests to a dedicated single machine in a constant number of \MPC rounds.

We can also assume, without loss of generality, that any single batch consists only of edge insertions or edge deletions. Indeed, since we consider constant-rounds algorithms, we can process the insertions first, and only then the deletions, resulting in a constant-rounds \MPC algorithm.

In our setting, we are assuming that the vertex set $V$ is fixed. It is rather easy to relax this requirement and allow insertions and deletions of isolated vertices, as long as a batch of updates can fit into a local machine. However, since the \MPC machines stay the same, we are assuming that even if the number of vertices changes, we will keep the same local memory \lspace.

\subsection{Related work}
\label{subsec:related}

There is a very rich literature covering dynamic graph algorithms and streaming graph algorithms; for some representative surveys, see, e.g., \cite{HHS22} and \cite{BFKKH23, CY20, McGregor14, Muthukrishnan05}, respectively.

Our results
can be compared to the state of the art of the problems studied in the streaming setting. The problem of connectivity and finding a spanning forest is known to have a $\Oh(n \log^3n)$ bits of space streaming algorithm \cite{AGM12}, matching our total memory bound in \Cref{thm:intro-con}. Nelson and Yu \cite{NY19} showed that $\Omega(n \log^3 n)$ bits of space is required for finding a spanning forest, implying that our total memory bound is tight. For finding an approximate minimum spanning forest, it is known (see, e.g., Ahn \etal \cite{AGM12}) that the space streaming complexity is essentially the same as connectivity, matching our bound in \Cref{thm:intro-mst}(i). For the problem of finding an $\alpha$-approximate matching, Assadi \etal \cite{AKLY16} (see also \cite{CCEHMMV16,Konrad15}) designed a streaming algorithm using space $\tOh(\max\{n^2/\alpha^3,n/\alpha\})$.
We are implementing this algorithm in \Cref{thm:intro-match}, matching this state-of-the-art complexity.
For insertions only, the upper bound of $\tOh(n/\alpha)$ is trivial and $\Omega(n/\alpha)$ is the lower bound to even store an $\alpha$-approximate matching in the worst case. Finally, for estimating the size of maximum matching, the state-of-the-art algorithms are due to Assadi \etal \cite{AKL17}, which use $\tOh(n^2/\alpha^4)$ space for general (turnstile) streams and $\tOh(n/\alpha^2)$ space for insertion-only streams, matching our bounds in \Cref{thm:intro-match}.

Algorithms in the \MPC model have been recently studied very extensively. While the early works on \MPC algorithms considered local memory\footnote{Notice that an \MPC with local memory $\lspace = \Oh(n)$ is almost the same as the Congested Clique model, see \cite{HP15,BDH18}.} $\Oh(n)$  or larger (see, e.g., \cite{ABBMS19,BHH19,CLMMOS18,GGKMR18,LMOS20,LMSV11}), recently the focus moved to the sublinear local memory. Some highlights include connectivity algorithms in $\tOh(\log D)$ rounds \cite{ASSWZ18,BDELM19}, study of some geometric problems \cite{ANOY14}, and $\tOh(\sqrt{\log n})$ rounds matching and MIS algorithms \cite{GU19,Onak18}.

\subsubsection{Related work in dynamic graph algorithms in the \MPC model}
\label{subsubsec:related-dynamic-algorithms}

Our model is very closely related to the D\MPC (\emph{Dynamic \MPC}) model introduced by Italiano \etal \cite{ILMP19} and later extended by Dhulipala \etal \cite{DDKPSS20} and Nowicki and Onak \cite{NO21}.

Italiano \etal \cite{ILMP19} were the first to emphasize the importance of the study of dynamic graphs in the \MPC framework, though their focus was on a single update. They obtained $\Oh(1)$-rounds \MPC algorithms handling a single update or query and maintaining connectivity, approximate MST, and maximal matching; the focus was on the local memory $\Oh(\sqrt{n+m})$ and the total memory was $\Oh(n+m)$. The emphasis was mainly on minimizing the total communication between different machines.

Dhulipala \etal \cite{DDKPSS20} extended the model to handle multiple edge updates and considered the case of sublinear local memory, while still obtaining constant-rounds \MPC algorithms, focusing on minimizing the total communication between different machines. The main result shows that on an \MPC with local memory $\lspace = \Oh(n^{\spacexp})$, in $\Oh(1)$ rounds one can maintain a dynamic undirected graph that can handle up to $k$ edge insertions/deletions and up to $k$ queries for connectivity between pairs of vertices, where $k n^{\eps} \polylog(n) \le \lspace$ for arbitrary constant $\eps>0$. The total communication for processing a batch of $k$ operations is $\tOh(k n^{\eps})$ and the total memory 
is $\tOh(n+m)$.

Nowicki and Onak \cite{NO21} continued this line of research and designed constant-rounds \MPC algorithms with local memory $\lspace = \Oh(n^{\spacexp})$ that can process $\Theta(\lspace)$ and $\Theta(\lspace^{1-\eps})$ updates to maintain minimum spanning forest and maximal matching, respectively. While the total memory used across all \MPC machines is $\tOh(n+m)$, this work has dropped the emphasis on the total communication between different machines. The focus was to maximize the size of the update batch processed while minimizing the number of rounds needed for the algorithm. As for communication, \cite{NO21} allows $\Theta(m)$ global communication during each round, which is usual in the static variant of the \MPC model, and is similar to our setup, except that we have $\tOh(n)$ global communication per round.

In regards to the model, the main difference between our setting and earlier related works
is that (i) \cite{ILMP19} considers only a single update, (ii) \cite{ILMP19,DDKPSS20} \emph{minimizes the overall number of messages sent} by \MPC machines in each round\footnote{In \cite{NO21}, the communication complexity is allowed to be $\Theta(m)$ global communication during each round, which is usual in the static variant of the \MPC model. The communication complexity in our algorithms is bounded by the respective total memory used, and hence it is $\tOh(n)$ in \Cref{thm:intro-con,thm:intro-mst,cor:intro-match-low-total}, which is substantially less than $\Oh(m)$ in the worst-case.}, and (iii) all papers \cite{ILMP19,DDKPSS20,NO21} use total memory $\widetilde{\Theta}(n+m)$ allowing to store the entire graph. In our paper, we want to
\begin{itemize}
\item \emph{maximize the number of updates} that can be processed and
\item \emph{minimize the total memory} used (aimed to be $\ll \tOh(n+m)$),
\item while \emph{minimizing the number of rounds} for processing (which we target to be constant).
\end{itemize}


\subsubsection{Distributed and \PRAM algorithms in dynamic setting}
\label{subsubsec:related-dynamic-PRAM-and-distributed-algorithms}

There has been only a limited amount of research on modeling dynamic parallel algorithms in modern distributed systems despite their potential impact on modern applications because of their speedup and better utilization of resources. Some early works include, e.g., Censor-Hillel \etal \cite{CHK16} who designed a dynamic algorithm for maintaining a maximal independent set of a graph in the \LOCAL model; the algorithm was maintaining only a single update though. See also \cite{AG18} etc.


The setting of parallel algorithms for multiple (batches of) updates has been recently investigated for the classical \PRAM and related models, see, e.g., \cite{AABD19,AABDW20,ABT20,DLSY21,DMVZ18,GL20,LSYDS22,TDS22}. The setting is similar to that studied in earlier works on \MPC \cite{ILMP19,DDKPSS20,NO21}, and in particular, it was not concerned with optimizing the total memory; the focus was on fast updates. The problems studied in these works include dynamic parallel algorithms for connectivity, spanning forest, MST, clique counting, etc.

\section{Technical overview}
\label{subsec:overview}


\subsection{Maintaining connectivity for a single update}
\label{subsubsec:overview-maintaining-connectivity-single-update}

The starting point of our connectivity algorithm is the streaming algorithm by Ahn, Guha, and McGregor (AGM)~\cite{AGM12} that uses $\Oh(n \log ^3 n)$ space. This algorithm can be implemented in the \MPC model where we can update every single insertion and deletion in $\Oh(1)$ rounds. However, to report the spanning forest, the algorithm requires $\Oh(\log n)$ rounds. The high-level idea of their algorithm is to maintain $\Oh(\log n)$ many independent sketches for each vertex.
At the end of the stream, one can merge the sketches of the vertices to generate a spanning forest in $\Oh(\log n)$ steps.

A direct \MPC implementation of the above AGM algorithm would require $\Oh(\log n)$ rounds to report a spanning forest as well. However, our goal is to spend only $\Oh(1)$ rounds per update and to report a spanning forest at any point in time. Towards this, we first present a streaming algorithm that uses $\Oh(n \log^3 n)$ space whose update time is $\Oh(n)$. Our idea is to maintain a spanning forest of the current graph at any instance along with the sketches of the vertices. When an edge $\{u,v\}$ is inserted such that $u$ and $v$ are in the same connected component, we need to only update the sketches of $u$ and $v$. If $u$ and $v$ are in different connected components, $\{u,v\}$ is included in the current spanning forest $F$ along with the update of the sketches. When an edge $\{u,v\}$ is deleted such that $\{u,v\}$ is not in $F$, we only need to update the sketches of $u$ and $v$. Otherwise, we consider the two components of $F$ that are created by deletion of the edge $\{u,v\}$ and let those be $Z_u$ and $Z_v$. We merge the sketches of the vertices in $Z_u$, and try to find a replacement edge between $Z_u$ and $Z_v$ which would be part of the spanning forest of the new graph. The detailed algorithm and analysis are presented in 
    \Cref{subsec:AGM12}
and \Cref{sec:correct}, respectively. To implement the above algorithm in \MPC, we rely on the technique of Euler tour trees. Each tree in $F$ is stored as an Euler tour and we consider the basic operations like rooting an Euler tour, merging two Euler tours, and splitting an Euler tour into two to facilitate the update operations in $\Oh(1)$ rounds. The details are presented in \Cref{algo:single-update}.

Next, we highlight the main differences between our algorithm and the AGM algorithm in \cite{AGM12}. 
While after each update, the data structure of the AGM algorithm can be updated in $\poly(\log n)$ sequential time, our algorithm 
requires \emph{$\widetilde{\Oh}(n)$ sequential time} to update its data structures. If we want to report a spanning forest of the current graph, the AGM algorithm requires $\Oh(\log n)$ steps ($\Oh(\log n)$ rounds in the \MPC model) while our algorithm explicitly stores a spanning forest of the current graph. The number of rounds in \MPC required to update the data structure in both the AGM algorithm and our algorithm is $\Oh(1)$. If we focus on reporting a spanning forest only at the end of the stream, then the AGM algorithm works as desired even if the number of updates is not bounded by $\poly(n)$. However, our algorithm requires the number of updates to be 
$\poly(n)$.

Now, we describe one of the possible main reasons why we can achieve $O(1)$ round complexity as opposed to $\Oh(\log n)$ rounds achieved by the AGM algorithm, in the \MPC model. Both our algorithm and AGM algorithm may need to merge $\Oh(n)$ sketches. Note that $\Oh(n)$ sketches can be merged in $\Oh(\log_{\lspace}n) = \Oh(1/\phi)$ rounds of \MPC. In the AGM algorithm, the sketches are not all just merged in one shot, but rather in $\Oh(\log n)$ steps. Merging of the components in a later round depends on the sketches used in previous rounds. Note that the AGM algorithm does not maintain any additional information apart from the sketches. Since our algorithm maintains the connected components along with the sketches, we take advantage of the extra information to bypass $\Oh(\log n)$ \MPC rounds barrier in direct implementation of AGM algorithm in our model.

\subsection{Maintaining connectivity for batch updates}
\label{subsubsec:overview-maintaining-connectivity-batches}

Now we describe how to extend the above algorithm for a large batch of updates that can be fit into the local memory of a single machine i.e., of size $\lspace=\Oh(n^\spacexp)$. For simplicity, we consider the insertions and deletions in a single batch separately in two rounds.

\paragraph{Insertions.}
Given a batch of insertions, for each edge, one of two cases must occur: either the edge is between two distinct components in $G$ and therefore it might become a part of the new spanning forest, or both endpoints of the edge are in the same component in $G$ which makes it a non-tree edge. The challenge here is to identify their types in one shot and update all the relevant information in a constant number of rounds. First, we construct an auxiliary graph $H$ where each vertex corresponds to a connected component of $G$. Insert all the edges to $H$ that do not create a self-loop or a parallel edge, and compute a spanning forest $F_H$ in $H$. This can be done in $\Oh(1)$ rounds. The edges in $F_H$ constitute the set of edges that merge two different connected components in $G$. The Euler tour for the updated graph is composed of the Euler tours of the existing components of $G$ via the edges of $F_H$. This amounts to determining the relative positions of the part of the trees in between two \emph{terminal} vertices (endpoint of an edge in $F_H$) in the final Euler tour. We first construct an Euler tour-like sequence $S$ for the auxiliary graph $H$
from which we show how to find an Euler tour for the whole graph in $\Oh(1)$ \MPC rounds. This is one of the main technical contributions of our work as previously, in the context of \MPC setting, efficient algorithms for maintaining Euler tours were only known under single edge updates. Updating the sketches and inserting the remaining edges (between the same connected components) extends from the single-edge update scenario.

\paragraph{Deletions.}
Now we discuss how we handle the deletions which is arguably the harder case. In a batch of deletions, for each edge, there are two cases to consider. If the edge is not part of the maintained spanning forest $F$ then the connectivity information and the spanning forest do not change. However, when an edge from $F$ is deleted either the corresponding component splits into two parts or we might find a \emph{replacement edge} that might connect back these two components. However, notice that we do not have all such replacement edges stored explicitly in the memory. For the deletion of a batch of tree edges, the challenge is to then identify all the potential replacement edges \emph{at once} and build the new spanning forest using those. We first delete the edges from the graph and update the sketches accordingly. To update the Euler tours, which are now split into many parts after the deletions, we employ an inverse procedure of the methods that we use to update the Euler tours for processing insertions. Next, we construct an auxiliary graph $H$ where each vertex corresponds to a connected component of $G$ that has at least one endpoint of a deleted edge from the batch. Since the edges that are not part of the spanning forest $F$ are not explicitly stored, we use the maintained sketches to recover the replacement edges between the vertices in $H$. We then find a spanning forest $F_H$ in $H$ which can be done in $\Oh(1)$ rounds. Finally, we need to add the edges of the spanning forest $F_H$ to the intermediate spanning forest of $G$ (after the deletions) to find the new spanning forest. The corresponding Euler tours due to these insertions can be updated using our algorithm for processing insertions as well.

Our algorithm for handling the batch updates crucially differs from the earlier works in the following way. While Italiano \etal \cite{ILMP19} use Euler tours to implement their connectivity algorithm for single-edge updates, implementing the Euler tour trees for batch updates is more technically challenging if we want to do so in $\Oh(1)$ rounds. While Nowicki and Onak \cite{NO21} consider batch updates of size $\Oh(n^\spacexp)$ in $\Oh(1)$ rounds for MST and 2-edge connectivity, they use the top tree data structure to implement their algorithms, which is arguably more complicated. Dhulipala et. al \cite{DDKPSS20} use parallel batch-dynamic tree data structure to process a batch of $\Oh(n^{1-\Theta(1)})$ updates in $\Oh(1)$ rounds for maintaining connectivity. We consider the parallel merging of a batch of Euler tour trees and the splitting of an Euler tour tree into a number of Euler tour trees in $\Oh(1)$ rounds, which is technically interesting. The batch size considered in our work is $\Oh(n^\spacexp/\poly(\log n))$. The loss of a $\poly(\log n)$ factor in the batch size is caused by storing $\Oh(\log n)$ many sketches. In order to process a batch of a certain size, all the sketches corresponding to the updates must fit onto a single machine. Dhulipala et. al \cite{DDKPSS20} also uses sketching but stores $n^{\Theta(1)}$ sketches per vertex. While  Dhulipala et. al \cite{DDKPSS20} focuses on minimizing the total communication along with achieving $\Oh(1)$ rounds update complexity, our objective is to achieve $\Oh(1)$ rounds algorithm while maximizing the batch size, similar to Nowicki and Onak \cite{NO21}. The crux of our algorithm is the use of total memory strictly sublinear in the graph's size, matching the space complexity of the best-known streaming algorithm for connectivity.

\subsection{Application of connectivity: Minimum spanning forest and bipartiteness}
\label{subsubsec:overview-MSF-bipartiteness}

We present applications of the connectivity result: (i) Exact Minimum Spanning Forest (MSF) in insertion-only streams, (ii) $(1+\eps)$-approximation to the weight of MSF, and (iii) testing whether the input graph is bipartite, both in dynamic streams. The total memory of all three algorithms is $\tOh(n)$ and the corresponding data structures in all three cases can be updated in $\Oh(1)$ rounds of \MPC.

For exact MSF in insertion-only streams, we consider the following (folklore) streaming algorithm that requires $\Omega(n)$ time to update the data structure for a single insertion. We maintain an MSF $F$ of the current graph at any instance of time. When an edge $e=\{u,v\}$ is inserted, we check whether $e$ is the heaviest edge in the cycle formed by joining $e$ with the path between $u$ and $v$ in $F$. We implement this algorithm in the \MPC model using Euler tours. Along with Rooting, Join, and Split operations on Euler tours, we consider another operation called \emph{Identify-Path} in which we are given two vertices, and the objective is to identify all the edges in the unique path in $F$ between them. We show that a batch of $\Oh(n^\spacexp)$ such Identify-Path operations can be executed in $\Oh(1)$ rounds which is the crux of our algorithm for exact MSF in insertion-only streams showing $\Oh(1)$ \MPC rounds are enough to update a batch of $\Oh(n^\spacexp)$ insertions. We also note here that our connectivity algorithm for batch updates is used here only in a black box manner.
The $(1+\eps)$-approximation to the weight of the MSF and testing whether the input graph is bipartite are more straightforward applications of our connectivity algorithm under batch updates.

\subsection{Approximate maximum matching}
\label{subsubsec:overview-approx-max-matching}

As already mentioned in \Cref{intro}, we give separate algorithms for finding an $\Oh(\alpha)$-approximate maximum matching and $\Oh(\alpha)$-approximation to the value of the maximum matching.

\paragraph{Matching finding.}
In insertion-only streams, we consider the folklore streaming algorithm that stores a matching greedily over the stream of size at most $\Oh(n/\alpha)$. We show that this algorithm can be implemented in \MPC such that $\Oh(1)$ rounds is enough to process an update of size at most $\Oh(n^{\spacexp})$. For dynamic streams, our approach relies on extending the known algorithms for finding an approximate matching and implementing those in the \MPC model suitably. However, it is not the case that all (sketching-based) algorithms in the streaming literature (to find an approximate matching) can be extended to \MPC. In particular, there are three different papers~\cite{AKLY16, CCEHMMV16, Konrad15} that give essentially the same result: there exists a dynamic streaming algorithm that uses $\widetilde{\Oh}(n^2/\alpha^3)$ space and reports a matching whose size is an $\Oh(\alpha)$ approximation to the size of the maximum matching. However, our techniques rely on the streaming algorithm of Assadi \etal \cite{AKLY16} and we do not see how the algorithms of~\cite{CCEHMMV16, Konrad15} can be extended to the \MPC setting such that we can efficiently process a batch of updates. At a high level, the algorithm by Assadi \etal \cite{AKLY16} generates a sparsified graph $H$ of size $\tOh(\max\{n^2/\alpha^3,n/\alpha\})$ such that any maximal matching of $H$ is an $\Oh(\alpha)$-approximation to the maximum matching of the input graph. We show that we can dynamically process a batch of updates in $\Oh(1)$ rounds to the input graph and generate a batch of updates to the graph $H$. Here, we use the \MPC algorithm by Nowicki and Onak~\cite{NO21} for maintaining a maximal matching that processes a batch of updates in $\Oh(1)$ rounds and uses total memory proportional to the size of the graph.

\paragraph{Matching size estimation.}
We build on the streaming algorithms by Assadi, Khanna, and Li~\cite{AKL21} that can report an $\Oh(\alpha)$ approximation to the size of the maximum matching; that uses $\widetilde{O}(n/\alpha)$ space and $\widetilde{O}(n/\alpha^2)$ space in insertion-only streams and in dynamic streams, respectively. The streaming algorithms of \cite{AKL21} for insertion-only and dynamic streams are based on a meta-algorithm (whose $\Oh(\log n)$ instances are run in parallel with different parameters) called $\mbox{\textsc{Tester}}(G,k)$, $k \in \mathbb{N}$. $\mbox{\textsc{Tester}}(G, k)$ takes a graph $G$ as input over a stream along with the parameter $k$ and distinguishes between $\mbox{OPT} \geq k$ and $\mbox{OPT} \leq  k/2$, where $\opt$ denotes the size of the maximum matching in $G$. Note that the space complexity of $\mbox{\textsc{Tester}}(G,k)$ is $\tOh(k)$ and $\tOh(k^2)$ in insertion-only and dynamic streams, respectively. In insertion-only, $\mbox{\textsc{Tester}}(G,k)$ maintains a matching of size at most $k$ greedily. In dynamic streams, $\mbox{\textsc{Tester}}(G,k)$ generates a sparse subgraph $H$ of size $\tOh(k^2)$. We show that we can dynamically process a batch of updates in $\Oh(1)$ rounds to the input graph and generate a batch of updates to the graph $H$. Here, we use the \MPC algorithm by Nowicki and Onak~\cite{NO21} for maximal matching and process a batch of updates in $\Oh(1)$ rounds.

\subsection{Organization}
\label{subsec:organization}

After some preliminaries in \Cref{sec:prelims}, we discuss a streaming algorithm for connectivity in dynamic streams in \Cref{sec:pure-streaming}. In \Cref{sec:connect-single}, we implement the above algorithm in \MPC under single edge updates. We conclude with some open directions in \Cref{sec:conclusion}. Please refer to the full version for the missing details.
\section{Preliminaries}
\label{sec:prelims}

Let $[n]$ denote the set $\{1,\ldots, n\}$. Let $G = (V(G),E(G))$ denote an \emph{undirected} graph $G$ with vertex set $V(G)$ and edge set $E(G)$. When $G$ is clear from the context, we may denote $V(G)$ and $E(G)$ by $V$ and $E$, respectively. 
Throughout the paper, $n$ denotes the number of vertices in $G$. For an edge $e$, $G \cup \{e\}$ denotes the graph with vertex set $V(G)$ and edge set $E(G) \cup \{e\}$. $G \setminus \{e\}$ denotes the graph with vertex set $V(G)$ and edge set $E(G) \setminus \{e\}$. A subgraph $F$ of a graph $G$ is called a \emph{spanning forest} if it is a collection of disjoint trees that cover all the vertices of $G$ without forming cycles. An edge present in a spanning forest is referred to as a \emph{tree edge} and all other edges in the graph are called \emph{non-tree edges}. For a vector $\mathbf{X}$, $\mathbf{X}_j$ denotes the $j$-th coordinate of $\mathbf{X}$ and $\ell_0(\mathbf{X})$ is the standard $\ell_0$ norm of vector $\mathbf{X}$, i.e., the number of non-zero elements in vector $\mathbf{X}$. When we say an event holds \emph{with high probability} (w.h.p.), we mean that the event holds with probability at least $1-1/\poly(n)$.

\subsection{Preliminaries on sketching}
\label{subsec:sketching}

Graph sketching provides a powerful approach to summarize and approximate the structural properties of graphs without needing to store or process the entire graph. Graph sketching algorithms typically operate incrementally as the graph stream arrives, updating the sketch accordingly. By constructing compact graph sketches, it becomes possible to approximate key graph properties.

First, in the following lemma, we discuss the $\ell_0$ sketching result of a vector in streaming. Then we use it to discuss sketching the neighborhood of the vertices in a graph.
\begin{lemma}[\cite{CJ19}]
\label{lem:ell0}
Let $\delta \in (0,1)$.
{There exists a dynamic streaming algorithm that receives updates to a vector $\mathbf{X}\in \{-1,0,1\}^{N}$, stores a sketch $S$ of $\Oh(\log ^2 N \cdot \log (1/\delta))$ bits for $\mathbf{X}$, and works as follows with probability $1-1/\delta$: upon a query to sketch $S$ we either get $\perp$ or an index $i \in [N]$ randomly from the set $\{j \in [N]:\mathbf{X}_j\neq 0\}$ depending on whether $\ell_0(\mathbf{X})=0$ or not.}
\end{lemma}
\begin{remark}
Note that the sketch in \Cref{lem:ell0} is \emph{linear} or \emph{mergable} i.e., let $S_1$ and $S_2$ be the linear sketches of vectors $\mathbf{X}_1$ and $\mathbf{X}_2$, then $S_1+S_2$ is also a sketch of $\mathbf{X_1}+\mathbf{X_2}$. This makes the updates in the streaming algorithm of \Cref{lem:ell0} easy.
\end{remark}
Using \Cref{lem:ell0}, we discuss the sketching results we will be using in this paper. To begin with, let $G$ be a graph with $V(G)=\{v_1,\ldots,v_n\}$.  For a vertex $v_i \in V(G)$, consider a vector $\mathbf{X}_{v_i}$ in $\{-1,0,1\}^{n \choose 2}$ whose each entry is of the form $\{j,k\} \in {[n] \choose 2}$ such that
\[ \mathbf{X}_{v_i} (\{v_j,v_k\})=\begin{cases}
      +1, &  \{v_j,v_k\}\in E(G)~\mbox{and}~i=\max\{j,k\}   \\
      -1, &  \{v_j,v_k\}\in E(G)~\mbox{and}~i=\min\{j,k\}   \\
      0, & \mbox{otherwise.}
   \end{cases}
\]


For a subset $S \subseteq V(G)$, let us define $\mathbf{X}_S:=\sum\limits_{v \in S}\mathbf{X}_v$. Let $E(A, B)$ be the set of edges having one endpoint in $A$ and the other endpoint in $B$.

Now, consider the following lemma that gives a connection between the vector $\mathbf{X}_A$ and the set of edges in $E(A, V\setminus A)$, where $A$ is a subset of $V(G)$.
\begin{lemma}[\cite{AGM12}]
Let us consider $A \subseteq V(G)$. Then $\size{E(A,V \setminus A)}=\ell_o(\mathbf{X}_A)$. Also, for any $\{j,k\} \in {[n] \choose 2}$, $\mathbf{X}_{A}(\{j,k\})=\pm 1$ if and only if $ \{j,k\} \in E(A,V\setminus A)$.
\end{lemma}
Next, we consider three different lemmas that talk about sketching the neighborhood of a vertex (\Cref{lem:sketch1}), or a set of vertices (\Cref{lem:sketch}), or between two sets of vertices (\Cref{lem:sketch-A-B}). All three lemmas follow from \Cref{lem:ell0}. However, we have stated them in different forms that will be useful in our proofs.
\begin{lemma}
\label{lem:sketch1}
{Let $v$ be a vertex of $G$ and $\delta \in (0,1)$. There exists a one-pass streaming algorithm that  stores an $\Oh(\log^2 n \cdot \log (1/\delta))$ size sketch $S_{v}$ of $\mathbf{X}_v$ such that the algorithm does the following with probability $1-\delta$:}
\begin{itemize}
\item when $E(\{v\}, V \setminus \{v\}) =\emptyset$,  i.e., $v$ is a singleton vertex, we get  $\perp$;
\item  $E(\{v\}, V \setminus \{v\}) \neq \emptyset$, we get  a random edge from $E(\{v\}, V \setminus \{v\})$, i.e., a random edge with $v$ as an endpoint.
\end{itemize}
\end{lemma}

\begin{proof}
This follows from \Cref{lem:ell0} where our objective is to store a sketch of the vector $\mathbf{X}_v$. Note that any update to the neighborhood of $v$ can be thought of as an update to the vector $\mathbf{X}_v$.
\end{proof}

\begin{lemma}\label{lem:sketch}
{Let $\delta \in (0,1)$.
There exists a one-pass streaming algorithm that stores an $\Oh(\log^2 n \cdot \log (1/\delta))$ size  sketch $S_{v}$ of $\mathbf{X}_v$ for each $v \in V(G)$ such that the algorithm does the following for a given subset $A \subset V(G)$ with probability $1-\delta$:}
\begin{itemize}
    \item reports $\perp$ when $E(A, V \setminus A)=\emptyset$
    \item reports a random edge from $E(A, V\setminus A)$ when $E(A, V \setminus A) \neq \emptyset$.
\end{itemize}
\end{lemma}
\begin{proof}
Let us consider the sketch $S_v$ for $\mathbf{X}_v$ as guaranteed by the algorithm \Cref{lem:sketch1}. Clearly, the space complexity of the algorithm is $\Oh(n \log ^3 n)$. For $A \subseteq V(G)$, let us consider the sketch $S_A= \sum_{v \in A}S_v$. Recall that $S_v$ is the sketch for $\mathbf{X_v}$, $\mathbf{X}_A=\sum_{v \in A}\mathbf{X}_v$, and sketches are linear. So, $S_A$ is a sketch for $\mathbf{X}_A$ such that the desired properties hold (as mentioned in the statement) w.h.p.
\end{proof}


\begin{lemma}\label{lem:sketch-A-B}
Let $A$ and $B$ be disjoint subsets of $V(G)$.
There exists a one-pass streaming algorithm that  stores an $\Oh(\log^3 n)$ size  sketch for $E(A,B)$ such that w.h.p. the algorithm does the following:
\begin{itemize}
    \item reports $\perp$ when $E(A, B)=\emptyset$
    \item reports a random edge from $E(A, B)$ when $E(A, B) \neq \emptyset$.
\end{itemize}
Note that the space complexity of the algorithm is $\Oh(n \log^3 n)$.
\end{lemma}
\begin{proof}
   {Follows from \Cref{lem:ell0} and setting $\delta=1/n^{\Oh(1)}$.}
\end{proof} 
\section{Streaming algorithm for spanning-forest}
\label{sec:pure-streaming}

Ahn, Guha, and McGregor~\cite{AGM12} proposed a streaming algorithm for maintaining a spanning forest that uses $\Oh(n \log ^3 n)$ bits of space and succeeds w.h.p. Moreover, the update time of the algorithm is $\poly (\log n)$. This algorithm can be implemented in \MPC where we can update every single insertion and deletion in $\Oh(1)$ rounds. But to report the spanning tree the algorithm will require $\Oh(\log n)$ rounds. Here, we first present a streaming algorithm that uses $\Oh(n \log^3 n)$ space whose update time is $\widetilde{\Oh}(n)$. However, as we show in \Cref{sec:connect-single}, this algorithm can be implemented in \MPC such that both update and query can be performed in $\Oh(1)$ rounds.

\subsection{High-level description of Ahn \etal \cite{AGM12}}
\label{subsec:AGM12}

The algorithm of Ahn \etal \cite{AGM12} maintains $t=\Oh(\log n)$ independent sketches
(cf. Lemma~\ref{lem:sketch1})
of $\mathbf{X}_v$ for each vertex $v$. Let $S_v^{(i)}$ be the $i$-th sketch of $\mathbf{X}_v$. Note that the total space complexity 
is $\Oh(n \log ^4 n)$. After the end of the stream, the algorithm finds a spanning forest of the graph in $t$ iterative steps as follows. In the first round, we find an edge for each vertex $v$ from sketch $S^{(1)}_v$ (if exists) and merge the vertices present in the same connected component in the graph $H_1(V_1, E_1)$, where $V_1=V(G)$ and $E_1$ is the edge set found in the first round. Let $V_2$ be the set of supernodes where each supernode corresponds to a connected component in $H_1$. Similarly, in the $i$-th round with $2 \leq i \leq t$, $V_i$ is the set of supernodes, each supernode corresponding to a connected component in $H_{i-1}$. We use the $i$-th sketches of the vertices to find edges from the supernodes in $V_{i}$ (if exists). Let $E_i$ be the set of edges between the nodes in $V_i$ found from the sketches. Observe that the algorithm terminates after $t=\Oh(\log n)$ rounds at which point each supernode in $V_t$ corresponds to a connected component of $G$.

The above algorithm can be implemented on an \MPC when the number of updates is $\poly(n)$. Note that, for updating the data structure with insertion 
or deletion of an edge $\{u,v\}$, we just need to update the sketches $S_{u}^{(1)},\ldots, S_u^{(t)}, S_v^{(1)},\ldots, S_v^{(t)}$. These updates to sketches can be done in $\Oh(1)$ rounds. To answer a query, we can compute a spanning forest of the current graph in $\Oh(\log n)$ rounds in the same way as we described above for finding a spanning forest at the end of the stream.

\subsection{Our algorithm using $\Oh(n \log^3 n)$ space}
\label{algo:single-update}

We start with the description of the data structure that our streaming algorithm maintains followed by the complete algorithm.

We maintain the following data structures:
\begin{itemize}
\item \textbf{Component Id}: For each $i \in [n]$ a value $C[i]$ that denotes the component id in which vertex $i$ lies. Initially, $C[i]=i$ for each $i \in [n]$. For a vertex $u \in V$, we denote by $C_u$ the component having vertex $u$, i.e., the set of vertices with component id $C[u]$. For simplicity, the component id of a component $C$ is the minimum id of any vertex in $C$.

\item \textbf{Spanning Forest:} A spanning forest $F$ of the current graph~$G$.

\item \textbf{Sketches:} For every vertex $v \in V(G)$, a sketch $S_v$ for $\mathbf{X}_v$.
%

\end{itemize}

\paragraph{Algorithm.}
Using the data structures defined above we can describe our algorithm in detail. 
The main algorithm is \textsc{Connectivity} that initializes the data structures and uses subroutines \textsc{Insert}, \textsc{Delete}, and \textsc{Query}.
\vspace{5pt}

\begin{algorithm}[H]
\caption{\textsc{Connectivity} $(e=\{u,v\})$}
\label{alg:conn}

\textbf{Initialization:} The spanning forest consists of $n$ components where each component corresponds to a vertex of the graph.  No sketch for the components is stored and those will be initiated over updates.

\textbf{Operation:} Based on the instruction whether insertion of an edge, deletion of an edge, or reporting of the current spanning tree call \textsc{Insert}, \textsc{Delete}, or \textsc{Query}; respectively.
\end{algorithm}


\paragraph{Insertion:}
We now describe the subroutine \textsc{Insert}  that updates the data structures when we insert an edge $e=\{u,v\}$. It updates the sketches corresponding to the vertices $u$ and $v$ and updates the spanning forest and the component ids of the vertices depending on whether $u$ and $v$ are in the same component currently.
\vspace{5pt}

\begin{algorithm}[H]
\caption{\textsc{Insert} $(e=\{u,v\})$  }
\label{algo:insert}
   Update the sketches $S_{u}$ and $S_v$ with the information that $\{u,v\}$ is inserted.

\If{$C[u]=C[v]$}
{
No change to $F$.

}
\Else
{
    Add  $\{u,v\}$ to $F$. \ {\small \textit{// Note that components $C[u]$, $C[v]$ are merged to form a bigger component}}

    Update the component id of all vertices in $C_u$ or $C_v$ as follows:

    \lIf{$C[u] < C[v]$}{the components ids of all the vertices in $C_v$ are changed to $C[u]$}
    \lElse{the component ids of all the vertices in $C_u$ are changed to $C[v]$}

}
\end{algorithm}

\paragraph{Deletion:}
Next, we describe subroutine \textsc{Delete}  that updates the data structures when we delete an edge $e=\{u,v\}$. It updates the sketches corresponding to the vertices $u$ and $v$ and updates the spanning forest and the component ids of the vertices depending on whether $\{u,v\}$ is an edge in the current spanning forest $F$ or not. When $\{u,v\}$ is not in $F$, we do not need to do anything. Otherwise (when $\{u,v\}$ is in $F$), the update procedure is slightly more involved. The component of $F$ having both $u$ and $v$ (before the deletion of $\{u,v\}$) will be divided into two components (after the deletion of $\{u,v\}$). Let $Z_u$ and $Z_v$ denote the set of vertices in the components having $u$ and $v$, respectively. The high-level idea is to merge the sketches of the vertices in $Z_u$ in order to find a possible edge from a vertex of $Z_u$ to $Z_v$. The details are in \Cref{alg:delete}.
\vspace{5pt}

\begin{algorithm}[H]
\caption{\textsc{Delete} $(e=\{u,v\})$}
\label{alg:delete}

// \textit {\small In this case $C[u]=C[v]$}

  Update the sketches $S_{u}$ and $S_v$ with the information that $\{u,v\}$ is deleted.

\If {$\{u, v\}$ is not in $F$}
{
   No need to do anything.
}
\ElseIf{$\{u,v\}$ is in $F$}
{
   Determine the sketch ${S}_{Z_u}$ of ${X}_{Z_u}$ by merging the sketches of all the vertices in $Z_u$.

    Use the sketch ${S}_{Z_u}$ to find a possible replacement edge between $Z_u$ and $Z_v$.

    \If {$S_{Z_u}$ reports $\perp$}
    {
        {\small \textit{ // In this case, there is no edge between $Z_u$ to $Z_v$}.}

        Update the component ids of the vertices in $Z_u \cup Z_v$ as follows:

        \quad $\triangleright$ all vertices in $Z_u$ are assigned a component id $\min_{x \in Z_u} C[x]$ and all vertices in $Z_v$ are assigned a component id $\min_{x \in Z_v} C[x]$.
        \label{algo:delete-line-id}
    }
    \ElseIf{$S_{Z_u}$ reports an edge $\{a,b\}$}
    {
        Add $\{a,b\}$ to $F$.
    }
}
\end{algorithm}

\medskip

The subroutine \textsc{Query} of the main algorithm is simple as we store an explicit spanning forest in our data structure.

\medskip

\begin{algorithm}[H]
\caption{\textsc{Query}}  \label{alg:query}
Report the spanning forest $F$ stored by the data structure.
\end{algorithm}

\subsection{Correctness proof of \textsc{Connectivity}}
\label{sec:correct}

We begin our analysis of Algorithm \textsc{Connectivity} with the following simple lemma.

\begin{lemma}
\label{lemma:space-Connectivity}
The space complexity of \textsc{Connectivity} is $\Oh(n \log^3 n)$.
\end{lemma}

\begin{proof}
From the description of \textsc{Connectivity}, our algorithm stores a spanning forest, a component id array $C$, and a sketch for each vertex $v \in V(G)$. By \cref{lem:sketch1}, $\Oh(\log^3 n)$ space is enough to store the sketch of a single vertex. Hence, the claimed space complexity bound follows.
\end{proof}

Next, we prove that Algorithm \textsc{Connectivity} maintains a spanning tree of the current graph $G$ whp.

\begin{lemma}
\label{lemma:correctness-Connectivity}
At any instant of time, Algorithm \textsc{Connectivity}  behaves as follows with high probability: $F$ is a spanning forest of (the current graph) $G$ and all vertices in a connected component $X$ have the same component id, i.e., $\min_{v \in X} C[v]$.
\end{lemma}

\begin{proof}
For the clarity of presentation, we use the normal notations before the update and dashed notations after the update i.e., after an update: graph $G$, spanning forest $F$, array $C$, and the sets $C_u$ and $Z_u$ are denoted by $G'$, $F'$, $C'$, $C'_u$, and $Z'_u$, respectively.

We prove the lemma by induction on the number of updates. Note that the lemma holds at the beginning of the algorithm when none of the edges were inserted (or deleted). To prove the inductive step consider the case of insertion and deletion separately.

\paragraph{Insertion case:}
Let the current update be the insertion of an edge $e=\{u,v\}$. Consider the procedure \textsc{Insert}$(e)$ (to insert an edge $e$).

\vspace{1mm}
\noindent {\bf Case 1:} ($C[u] = C[v]$) In this case, there is no change to $F$ (i.e., $F'=F$) and note that $F$ is a spanning forest of $G$ as well as $G'=G \cup \{e\}$. That is, both arrays $C$ and $C'$ are identical. 

\vspace{1mm}
\noindent {\bf Case 2:} ($C[u] \neq C[v]$) In this case $\{u,v\}$ is added to the current spanning forest $F$ and the earlier components $C_u$ and $C_v$ are merged using the edge $\{u,v\}$ to form a bigger component. Observe that $F'=F \cup \{e\}$ is a spanning forest of $G'=G \cup \{e\}$. In \textsc{insert}, we update the component ids of the vertices in $C_u$ or $C_v$ to $C_v$ or $C_u$ depending on whether $C[u]> C[v]$ or not. Observe that, there is no change to any component except for $C_u$ and $C_v$.


\paragraph*{Deletion Case:} Let the current update be the deletion of an edge $e=\{u,v\}$. Consider the procedure \textsc{Delete}$(e)$ (to delete an edge $e$).

\vspace{1mm}
\noindent {\bf Case 1:} ($\{u,v\}$ is not an edge in $F$) In this case, there is no change to $F$ (i.e., $F'=F$) and note that $F$ is a spanning forest of $G$ as well as $G'=G \cup \{e\}$.  That is, both arrays $C$ and $C'$ are identical.  

\vspace{1mm}
\noindent {\bf Case 2:} ($\{u,v\}$ is an edge in $F$) In this case $\{u,v\}$ is deleted from the current spanning forest $F$. Before that, the sketches $S_u$ and $S_v$ are updated to take care of the deletion $\{u,v\}$. The algorithm finds the merged sketch $S_{Z_u}$ of $\mathbf{X}_{Z_u}$. We have the following observation:
      
\begin{observation}
\label{obs:Suv}
With high probability, the following holds about sketch $S_{Z_u}$: upon a query we either get $\perp$ or a random edge from $E(Z_u,Z_v)$ depending on whether $E(Z_u,Z_v) =\emptyset$ or not, respectively.
\end{observation}

\begin{proof}
By the induction hypothesis, w.h.p. $F$ is the spanning forest of the graph $G$ before the deletion of the current edge $e$. Recall that we assume the adversary is oblivious, i.e., the current edge to be deleted is not provided by an adversary rather it is fixed obliviously. So, the current edge update is independent of $Z_u$ and sketches. Hence, setting $A=Z_u$ in \cref{lem:sketch} completes the proof of \Cref{obs:Suv}.
\end{proof}

Algorithm \textsc{Delete} queries the sketch $S_{uv}$ after deletion of the edge $\{u,v\}$. We divide the analysis into two parts based on whether there is an edge between $Z_u$ and $Z_v$ or not.

If there is no edge between $Z_u$ and $Z_v$ (other than $\{u,v\}$), by \Cref{obs:Suv}, $S_{Z_u}$ reports $\perp$ with high probability.  Then observe that $F'=F \setminus \{e\}$ is a spanning forest of $G'=G \setminus \{e\}$. Apart from the components $C_u$ and $C_v$, the other components remain the same in both $G$ and $G'$. Also, the component ids of the vertices in $Z_u \cup Z_v$ are updated as required.

Now consider the case when there is an edge between $Z_u$ and $Z_v$ (other than $\{u,v\}$). By \Cref{obs:Suv}, $S_{Z_u}$ reports an edge $\{a,b\}$ between $Z_u$ and $Z_v$ with high probability. Assume that  $a \in Z_u$ and $b \in Z_v$. In algorithm \textsc{Delete}, $e'=\{a,b\}$ is added to $F \setminus \{e\}$. Observe that $F\setminus \{e\} \cup \{e'\} $ is a spanning forest of $G \setminus \{e\}$. Also, note that the component ids of the vertices remain the same as it was before the deletion of $e$, i.e., the arrays $C$ and $C'$ are identical.
This completes the proof of \Cref{lemma:correctness-Connectivity}.
\qedhere
\end{proof}





\section{MPC implementation for a single update}
\label{sec:connect-single}

In this section, we show how to implement the streaming algorithm from \Cref{sec:pure-streaming} 
in the \MPC model. We start with describing the Euler tour tree data structure 
that is central to our algorithm.

Let $T$ be a rooted tree. An \emph{Euler tour} (\emph{E-tour} in short) $\text{EU}_T$ of $T$ is a walk along $T$ that starts and ends at the root, visits each edge exactly twice, and is represented by a sequence of the endpoints of the traversed edges. E.g., if the path uses the edges $\{u,v\}$ and $\{v,w\}$, then $v$ appears twice in the sequence; each vertex $v$ appears $2d_T(v)$ times in the sequence, where $d_T(v)$ is the degree of $v$ in $T$.


The tree corresponding to an E-tour is called an E-tour tree. The length of the E-tour $\eu_T$ of a E-tour tree $T$ is denoted by $L_T=4(|T|-1)$. Let $\eu_{T_u}$ denote an E-tour having $u$ as one of the vertices. Let $f(u)$ and $\ell(u)$ denote the index of the first and the last occurrence of $u$ in $\eu_T$. Let $\ind _u$ denote the set of all indices where vertex $u$ occurs in $\eu_T$. Note that $|\ind_v| = 2 \cdot  d_T(v)$ in the E-tour, where $d_T(v)$ is the degree of $v$ in the 
E-tour tree $T$.

Our goal is to maintain the E-tours in \MPC and execute on them some basic operations. Let us consider a graph $G=(V(G), E(G))$ and a spanning forest $F$ of $G$ such that the spanning tree of each component of the graph is rooted. For each (rooted) spanning tree $T$ of forest $F$, we maintain an Euler tour $\eu_T$. Using a vertex-based partitioning (with all edges incident to a vertex stored on consecutive machines), we distribute the edges across all machines. For each vertex $v$, we maintain:
%
%
\begin{itemize}
\item the id of its component or the id of the E-tour in which it is present; 
\item the index of its first and last appearance in the E-tour, i.e., $f(v)$ and $\ell(v)$;
\item $\ind _v$ (the set of all indices in which $v$ appears in the E-tour of $T$). We do not explicitly store $\ind_v$. This is implicitly stored as information on the edges incident on $v$.
\end{itemize}

Furthermore, for each edge $\{u,v\}$, we store:
\begin{itemize}
    \item the id of the E-tour having both $u$ and $v$;
    \item the index of first and last appearance of $u$ and $v$ in the E-tour, i.e., $f(u), f(v),\ell(u)$, and $\ell(v)$;
    \item the index in the E-tour where $u$ and $v$ appear such that those indices correspond to traversing the edge $\{u,v\}$.
\end{itemize}

To update an E-tour or to perform some operation implies changing the information stored with the edges. Note that $\ind_u$ for vertex $u$ is stored in a distributed manner together with the edges incident on $u$ from which $f(u)$ and $\ell(u)$ can be computed. So, to update the E-tours in a distributed fashion, we leverage the properties of an E-tour which allows us to perform basic operations like \emph{Rooting}, \emph{Join}, and \emph{Split} by communicating only $\Oh(1)$ size messages. 

\subsection*{Operations on Euler tours}
   \noindent {\bf Rooting:} Given an E-tour $\eu_T$ of a tree $T$ and a vertex $u$ in $T$, the objective is to change $\eu_T$ to $\eu_{T'}$ such that $\eu_{T'}$ is an E-tour of tree $T'$, where $T'$ is a tree rooted at vertex $u$ having the same vertex and edge set as $T$.

    \vspace{1mm}
    \noindent  {\bf Join:} Given two E-tours $\eu_{T_u}$ and $\eu_{T_v}$ of trees $T_u$ and $T_v$ with roots $u$ and $v$, respectively, such that the vertex set of $T_u$ and $T_v$ are disjoint. Upon insertion of an edge $\{u,v\}$, $T_u$ and $T_v$ are combined into a single component. Let $T$ be the spanning tree of the union of the vertices in $T_u$ and $T_v$, obtained by combining $T_u$ and $T_v$ along the edge $\{u,v\}$. The goal is to generate the E-tour for the tree $T$.

    \vspace{1mm}
    \noindent  {\bf Split:} We are given an E-tour $\eu_T$ of a tree $T$ and an edge $\{u,v\}$ in $T$. Upon deletion of an edge $\{u,v\}$ in $T$, $T$ is divided into two components. Let $T_u$ and $T_v$ denote the spanning trees of the components having $u$ and $v$, respectively. The objective is to generate the E-tours for trees $T_u$ and $T_v$.


\begin{lemma}\label{lem:etour-oper}
The operations Rooting, Join, and Split
can be implemented on E-tours in $\Oh(1)$ rounds in \MPC.
\end{lemma}

\begin{proof}
Below we analyze each of these operations separately. For simplicity of presentation, we use $index_x, f(x),\ell(x)$ for a vertex to denote the respective quantity before the current update and $index'_x, f'(x),\ell'(x)$ to denote that of before the current update. The corresponding element of $i \in \ind_x$ (before the current update) will be denoted by $i' \in \ind_x'$ after the update.

   \vspace{2mm}
     \noindent {\bf Rooting:} Let the given E-tour be $\eu_T$ for tree $T$ and $u$ be the vertex such that we want to change the root of $T$ to $u$ (unless $u$ is the root of $T$ already). Recall that $L_T$ denotes the length of $\eu_T$. For each vertex $x$ in $T$ and $i \in \ind_x$, set $i'=(i+L_T-\ell(u))~\mbox{mod}~L_T +1$. These modifications of the indices can be done in $\Oh(1)$ rounds by broadcasting the id of $\eu_T$, the value of $L_T$, and $\ell(u)$ to all machines. After the modifications of the index sets of the vertices, the information on the edges corresponds to the new E-tour of the tree rooted at $u$.

      \vspace{2mm}
      \noindent {\bf Join:} Let the two E-tours given be $\eu_{T_u}$ and $\eu_{T_v}$ of trees $T_u$ and $T_v$ with roots $u$ and $v$, respectively, such that the vertex set of $T_u$ and $T_v$ are disjoint. Suppose we want to generate an E-tour for tree $T$ that can be formed by joining $T_u$ and $T_v$ with edge $\{u,v\}$ such that $u$ is the root of $T$. For each vertex $x \in T_v$ (including $v$), $\ind_x$ is modified as follows: for each $i \in \ind_x$  $i'=i+f(u)+2 \in \ind'_x$. Then $\ind_v$ and $\ind_u$ are modified as follows:  $\ind_v'=\ind_v \cup \{f(u)+2,f(u)+\ell(v)+3\}$ and $\ind_u'=\ind_u \cup \{f(u)+1,f(u)+\ell(v)+4\}$. Here $\ell(u) \in \ind_u$ denotes the last index of $u$ in the E-tour of $\eu_{T_u}$ before the insertion of $\{u,v\}$. Then for each $i \in \ind _x$ for some $x$ in $T_u$ such that $i > f(u)$, $i'=i+L_{T_{v}}+4 \in \ind_x'$.
       These modifications of the indices can be done in $\Oh(1)$ rounds by broadcasting the IDs of $\eu_{T_u}$ and $\eu_{T_v}$, the values of $L_{T_u}$ and $L_{T_v}$, and the values of $f(u),\ell(u),f(v)$, and $\ell(v)$ to all machines. After the modifications of the index sets of the vertices, the information on the edges corresponds to the new E-tour of the tree $T$.

        \vspace{2mm}
       \noindent  {\bf Split:} Suppose we are given an E-tour $\eu_T$ of a tree $T$ and an edge $\{u,v\}$ in $T$.  Suppose we want to split $\eu_T$ into two E-tours for two spanning trees after deleting edge $\{u,v\}$: one for tree $T_u$ (containing the vertex $u$) and the other one is for $T_v$ (containing the vertex $v$). First, we decide whether $u$ is an ancestor of $v$ or $v$ is an ancestor $u$ from the value of $f(u),f(v), \ell(u)$, and $\ell(v)$. In particular $u$ is an ancestor of $v$ if and only if $f(u)< f(v)$ and $\ell(u)>\ell(v)$. Without loss of generality, assume that $u$ is an ancestor of $v$. Modify $\ind_v$ and $\ind_u$ to $\ind_v \setminus \{f(v), \ell(v)\}$ and $\ind_u \setminus \{f(v)-1, \ell (v)+1\}$, respectively. For each descendent $x$ of $v$ (including $v$) and index $i \in \ind_x$, set   $i'=i-f(v)\in \ind'_x$. For each vertex $x$ in $T_u \setminus T_v$ and $i \in \ind_x$ such that $i>\ell(v)$, set $i'=i-(\ell(v)-f(v)+3) \in \ind_x'$. These modifications of the indices can be done in $\Oh(1)$ rounds by broadcasting the id of $\eu_{T}$ the value of $L_{T}$, and the values of $f(u),\ell(u),f(v)$, and $\ell(v)$ to all machines. After the modifications of the index sets of the vertices, the information on the edges corresponds to the new E-tours of the trees $T_u$ and $T_v$.
%
%
%
\end{proof}
%
With the aid of the Euler tour trees described above and the implementation of the operations on them, we can implement our streaming algorithm for \textsc{Connectivity} efficiently in \MPC. 
%
\begin{lemma}
\label{lem:analysis-single}
Each update and query in the algorithm in \Cref{algo:single-update} can be implemented in $\Oh(1)$~rounds.
%
\end{lemma}

\begin{proof}
The query step is trivial as our algorithm stores a spanning forest and the E-tour trees. We analyze the cases for insertions and deletions separately.

\vspace{1mm}
\noindent {\bf Insertion of an edge $\{u,v\}$:} As we are storing the sketches for each vertex and since the sketches are linear, the sketches of the vertices $u$ and $v$, i.e., $S_u$ and $S_v$ can be updated in $\Oh(1)$ rounds by just broadcasting the information that the sketches of $S_u$ and $S_v$ needs to be updated with insertion of edge $\{u,v\}$.
If $C[u]=C[v]$, our algorithm does not do anything apart from updating sketches.  Now consider the case $C[u]\neq C[v]$. In this case,  we want to merge the two components $C_u$ and $C_v$. 

\vspace{1mm}
\noindent {\bf Deletion of an edge $\{u,v\}$:} Note that $S_u$ and $S_v$ can be updated in $\Oh(1)$ rounds by just broadcasting the information that the sketches of $S_u$ and $S_v$ needs to be updated with deletion of edge $\{u,v\}$. If $\{u,v\}$ is not in $F$, our algorithm does not do anything apart from updating sketches. Now consider the case when $\{u,v\}$ is in $F$. We perform the Split operation for E-tours $\eu_{T}$ with both vertices $u$ and $v$ with the deletion of edge $\{u,v\}$. Let $T_u$ and $T_v$ be the two E-tours trees we obtain after the Split operation. We merge the sketches of the vertices in E-tour $\eu_{T_u}$ to get $S_{Z_u}$. This can be done by first broadcasting the id of the E-tour $\eu_{T_u}$ and then merging the sketches of the vertices in $Z_u$. Note that the broadcasting can be done in $\Oh(1)$ rounds and the merging step can be done in $\Oh(1/\spacexp)$ rounds\footnote{This is because each machine can store $\Omega(n^\spacexp / \log^3 n)$ sketches and hence we can generate the sketch $S_{Z_u}$ in $\Oh(\log_{n^\spacexp/\log ^3 n} n)=\Oh(1/\spacexp)$ rounds.}. Also, notice that the size of $S_{Z_u}$ is $\Oh(\log ^3 n)$ and can be stored in one machine. If we don't find any edge from the sketch $S_{Z_u}$, we are done. If we find an edge $\{a,b\}$ from the sketch $S_{Z_u}$, observe that it must be between $Z_u$ and $Z_v$ or equivalently one of $\{a,b\}$ is in $T_u$ and the other one is in $T_v$. Then we perform the Rooting operation to change the roots of $T_u$ and $T_v$ to $u$ and $v$, respectively. Then we execute Join operation to merge the E-tours $\eu_{T_u}$ and $\eu_{T_v}$ with insertion of edge $\{a,b\}$.  From \Cref{lem:etour-oper}, all steps required to delete edge $\{u,v\}$ can be done in $\Oh(1)$ rounds.

From the above descriptions along with the fact that all the operations Rooting, Join, and Split are in $\Oh(1)$ rounds (\Cref{lem:etour-oper}), we conclude that each update can be performed in $\Oh(1)$ rounds.
\end{proof} 
\section{Maintaining \textsc{Connectivity} under batch updates}
\label{sec:connect-batch}

In this section, we describe how our algorithm under single edge updates from the previous section can be extended to handle a batch of updates of size $\Oh(n^\spacexp)$ within the same (local and total) memory bounds as before while running in a constant number of rounds.

Suppose we are given a graph $G=(V, E)$ and a batch of $k = \Oh(n^\spacexp/{\log ^3 n})$ updates $U = \{e_1, e_2, \dots, e_k\}$ such that for each $e_i=(u_i,v_i)$ we have $u_i,v_i \in V(G)$. As before, we denote by $G'$ the graph $G$ after a batch of edges has been inserted into or deleted from $G$. For the sake of simplicity, we consider the insertions $I$ and deletions $D$ in a batch of updates $U = I \cup D$ separately in two consecutive phases. In the first phase, we handle the insertions $I$ that we describe next.

\subsection{Batch insertions}
\label{subsec:ins}

As before, after receiving a batch of edge insertions, our goal is to update the connectivity information of $G$ and also to maintain a spanning forest $F$ of $G$ explicitly. When a batch of insertions arrives, for each edge one of two cases must occur -- either the edge is between two distinct connected components in $G$ therefore it might become a tree edge in $G'$, or both endpoints of the edge are in the same connected component in $G$ which makes it a non-tree edge. The challenge is to identify their types in one shot and update all the relevant information in a constant number of rounds.

\paragraph{Data Structures.}
We maintain the same data structures as in \Cref{algo:single-update} here as well. In particular, we maintain (i) an array $C$ consisting of the component id of each connected component of $G$, (ii) a spanning forest $F$ of the current graph $G$, and (iii) a sketch $S_v$ of $\mathbf{X}_v$ for each vertex $v \in V(G)$. Furthermore, we also maintain the Euler tour tree data structure for the spanning forest $F$ as in \Cref{sec:connect-single}. In particular, we assume that the operations like Rooting, Join, and Split can be performed on E-tours in $\Oh(1)$ rounds.

\paragraph{Updating the Sketches.}
The set of inserted edges is denoted by $I \subseteq U$. We update the sketches independently for each inserted edge $e \in I$ similar to \Cref{sec:pure-streaming} i.e., for $e=\{u,v\}$ we update the sketches $S_u$ and $S_v$. As before, this can be achieved in $\Oh(1)$ rounds by broadcasting the information that the sketches of $S_u$ and $S_v$ need to be updated with insertion of edge $e = \{u,v\}$ for each $e \in I$. \\

Our main algorithm proceeds as follows. Consider the set of edges $I_{d} \subseteq I$ such that the endpoints of an edge $e =\{u,v\} \in I_{d}$ are in distinct connected components (i.e., $C[u] \neq C[v]$) in graph $G$ before the insertion. We will consider the remaining edges afterward. First, construct an auxiliary graph $H$ such that each vertex $v \in V(H)$ corresponds to a connected component of $G$ that contains an endpoint of some edge $e \in I_{d}$. Insert all the edges in $I_{d}$ between the corresponding (super) vertices in the graph $H$ and remove any parallel edges created in this process. Denote this set of left-over parallel edges by $I_{d'}$. Note that by construction, $H$ does not contain any self-loops either. Next, we compute a spanning forest $F_H$ in $H$.

\begin{claim}
\label{claim:ins-forest}
The spanning forest $F_H$ can be constructed in $\Oh(1)$ rounds.
\end{claim}

\begin{proof}
First notice that $|V(H)| \le 2k$ and $|E(H)| \le k$. For $k = \Oh(n^\spacexp/{\log ^3 n})$ the graph $H$ consisting of the components and the edges can be stored in the memory of a single machine. To construct $H$, each node in $V(H)$ can be identified from the array $C$. Since $H$ can be stored in a single machine, we can find $F_H$ in $\Oh(1)$ rounds as well.
\end{proof}

We first insert, in parallel, the edges of the spanning forest $F_H$ to the graph $G$. We have already updated the relevant sketches for the whole batch of insertions. We now describe how to update the corresponding E-tours after inserting the edges of $F_H$.

\subsection{Updating Euler Tour trees}
\label{subsec:update-Euler-tour-trees}

Our goal is to combine the E-tour trees corresponding to the vertices of the auxiliary graph $H$ using the edges of $E(F_H)$ to construct the E-tour of the spanning forest $F'$ of the updated graph $G'$. Wlog. suppose $|E(F_H)| = k$. Formally we can formulate the problem as follows: given $t \le 2k$ trees $T_1, \ldots, T_{t}$ and $k$ edges $e_1, \ldots, e_k$ such that the union of the trees $T_i$ and the edges $e_i$ form a forest $F'$ (in $G'$), our goal is to construct the Euler tour sequence of $F'$. For simplicity assume that the forest $F'$ is actually a tree $T$ and let the corresponding tree in $H$ be $T_H$. If $F'$ contains more than one tree then each of them can be handled in parallel in a similar fashion.

Define the root of this auxiliary tree $T_H$ to be the node $T_1$. Call the set of vertices that are endpoints of the edges $e_1, \ldots, e_k$ as \emph{terminals}. For each node $T_i$ in the tree $T_H$, the terminals of $T_i$ are the terminal vertices that are present in $T_i$.
Let us first root tree $T_1$ at an arbitrary terminal in $T_1$. For each non-root node $T_i$, let $t_i$ be the terminal vertex connected to some vertex in the parent of $T_i$. Notice that, $t_i$ might not be the root in the tree $T_i$ but we can change the root to $t_i$ by applying the Rooting procedure on $T_i$ (as explained in \Cref{sec:connect-single}) which can be done in $\Oh(1)$ rounds. Next, we define an auxiliary sequence $S$ (of edges from $E(F_H)$) for the entire spanning tree $T$ of the updated graph $G'$ from which we can construct the final E-tour $\eu_T$. We first create a node $T_0$ in $H$ such that $T_0$ consists of a single vertex $u_0$. Then connect $T_0$ to the existing root $T_1$ of $T_H$ through an edge $e_0=\{u_0,v_0\}$, where $v_0$ is the root of $T_1$.

\begin{definition}[Auxiliary Sequence]
\label{def:aux}
Let us consider the tree $T_H$ defined above where each node in $T_H$ corresponds to a component in $G$. The auxiliary sequence $S$ of $T_H$ is
 defined recursively using function $\Pi$ as follows.
\begin{itemize}
    \item For a leaf node $A$ in $T_H$, we have $\Pi(A) = \phi$.
    \item  For any other node $A$, consider the descendant nodes $B_1,\dots,B_p$ of $A$ in $T_H$. Let $u_1,\dots,u_p \in A$ and $v_1 \in B_1,\dots,v_p \in B_p$ such that $e_i=\{u_i,v_i\}$ are edges of $T_H$ and $f(u_1)\leq \dots \leq f(u_p)$ (where $f$ is defined in the E-tour corresponding to node $A$). Also, let $e_i'$ denote the same edge $e_i$ traversed in the other direction.
$$\Pi(A) = e_1 \Pi(B_1) e'_1 \dots {e}_p \Pi(B_p) e'_p.$$
\end{itemize}
Then we have $S =\Pi(T_0)$.
\end{definition}


\begin{lemma}
\label{lem:seq}
$S$ is a sequence of edges of length $\Oh(k)$ and can be constructed in $\Oh(1)$ rounds.
\end{lemma}

\begin{proof}
    Since the sequence $S$ is of total length $\Oh(k)$, we can construct it in $\Oh(1)$ rounds by applying the function $\Pi$ repeatedly inside the local memory of a single machine. Evaluating the function itself amounts to finding all the neighbors of a vertex in $T_H$ which can also be done inside the local memory of a single machine.
\end{proof}

Before proceeding further we first recall some notations from \Cref{sec:connect-single} and also introduce a few more. Every node $u \in V(G)$ has a first and a last index denoted by $f(v)$ and $l(v)$. The length of the E-tour of an E-tour tree $T_i$ is denoted by $L_{T_i}$. Interchangeably, we sometimes denote this by $L_A$ where $A$ is the node corresponding to $T_i$ in $T_H$. We denote by $\ind _v$ the set of all indices where vertex $v$ occurs in $\eu_T$. Recall that we do not explicitly store $\ind_v$ and it is only implicitly stored on the edges incident on $v$.

An Euler tour of a tree is a sequence of vertices on the tree. So, the Euler tree representation of a spanning forest a sequences of Euler tour trees one for each tree in the forest. When a batch of edges is inserted, the edges connect the existing trees through various (possibly internal) vertices. From the E-tours of the existing trees to get the new E-tour we need to compose the sequences appropriately. As mentioned before, we first root each tree $T_i$ for $i \in [t]$ at one of the terminals (endpoint of an edge in $T_H$), say $t_i$, and then construct the auxiliary sequence $S$. Now to convert the sequence $S$ to an E-tour of the tree $T$ we need to generalize the \emph{Join} operation to be able to merge several trees in $\Oh(1)$ rounds, instead of just two as in \Cref{sec:connect-single}.

For joining two E-tours $\eu_{T_u}$ and $\eu_{T_v}$ of trees $T_u$ and $T_v$, for (a subset of) vertices $x \in T_v$ and $y \in T_v$, the sets $\ind_x$ and $\ind_y$ are modified accordingly. The same holds for joining two E-tours, say $\eu_{T_v}$ and $\eu_{T_w}$. However, the main challenge to join multiple such E-tours together at once is to be able to figure out quickly, in the above scenario, how $\ind_z$ is modified for vertices $z \in T_w$ when the sets $\ind_x$ and $\ind_y$ are also modified and vice versa. Notice that if an E-tour sequence is inserted inside another E-tour sequence (which corresponds to attaching a tree to an internal vertex of another) then all the indices following the insertion need to be shifted appropriately. This was easy to achieve in the single edge insertion case as both the terminals can be made the root of the corresponding trees. However, this is not possible for multiple-edge insertions where a single tree can have many terminal vertices and only one can become the root.

We introduce two operations, \emph{shift-index} and \emph{update-index}, to overcome these challenges. We use shift-index to appropriately offset each interval between two terminal vertices in a tree for the new E-tour. Notice if a tree only contains a single terminal the indices for the whole tree would be shifted accordingly. Let us denote by $\Delta$ the quantity by which a set of indices are shifted. From the sequence $S$, we compute all the $\Delta$ values that we finally broadcast to each machine. Furthermore, for each inserted edge $e = \{u, v\}$, we use the update-index operation to add (in case of insertions) new elements to the index sets $\ind_u$ and $\ind_v$.

Next, we describe these operations in detail. We scan through the sequence $S$ and consider a consecutive pair of edges in $S$ at a time. The following cases may arise depending on whether the edge types are ``forward'' or ``backward'' which in turn depends on whether we are traversing the edge for the first time or not. To simplify the notation we assume that for each $e=\{u,v\}$, $u$ is closer to the root in the graph $H$.

\paragraph{Case 1:}
Suppose the next pair of edges in the sequence is both of forward type, say $e_1e_2$ where $e_1 = \{u_1, v_1\}$ and $e_2 = \{u_2, v_2\}$. In the auxiliary graph $H$, let $u_1$ belongs to node $A$, $v_1, u_2$ belong to node $B$, and $v_2$ belongs to  node $C$.

For each vertex $x \in B$ and $\forall z \in index_x$ such that $f(v_1) \le z \le f(u_2)$ we update $z' = z + \Delta$. Note that we are modifying only $z \in index_x$ such that $f(v_1) \le z \le f(u_2)$, which depends only on $\Delta$ as the terminals other than $v_1$ and $u_2$ of the E-tour of the node $B$ have higher $f$ values. For $z \in index_x$ such that $z > f(u_2)$ depends not only on the value of $\Delta$ but also on the length of the E-tour corresponding to $C$ and possibly the length of some other E-tours, and will be taken care of later either in Case 3 or Case 4.

As the size of the index of a vertex is twice the degree of the vertex in the spanning tree, two elements must be added to  $index_{u_2}$ and $index_{v_2}$ corresponding to the insertion of the edge $\{u_2,v_2\}$. But here we add one element each to $index_{u_2}$ and $index_{v_2}$ as follows.
$$ index'_{u_2} = index_{u_2} \cup \{ f(u_2) + \Delta + 1\} \qquad
index'_{v_2} = index_{v_2} \cup \{ f(u_2) + \Delta + 2\}.$$

Other addition of one element each to $index_{u_2}$ and $index_{v_2}$ (corresponding to the insertion of edge $e_2 = \{u_2,v_2\}$) will be taken care of in Case 2 or Case 4, as that depends on the length of the E-tour corresponding to $C$ and possibly the length of some other E-tours.

Finally, we update $\Delta' = \Delta + f(u_2) + 2$ to modify the value of the shift, as $\Delta$ was taking care of the shift earlier and the indices of the vertices in the descendants of $B$ (including $C$) will  experience a shift depending on $f(u_2)$.

\paragraph{Case 2:}
Suppose the next pair of edges in the sequence is of type forward followed by backward edge, say $ee'$ where $e = \{u, v\}$. In the auxiliary graph $H$, let $u$ belong to node $A$ and $v$ belong to node $B$. For each vertex $x \in B$ and $\forall z \in index_x$, we update $z' = z + \Delta$. Note that  we are modifying $\forall z \in index_x$ where $x \in B$, as this depends only on $\Delta$ and not on the length of the other E-tours in Case 1. This is because $B$ is a leaf node in $H$ and has no terminal node other than $v$.

We also update $index_{u}$ and $index_{v}$ as:
$$index'_{u} = index_{u} \cup \{ L_{B} + \Delta + 2\}
\qquad
index'_{v} = index_{v} \cup \{ L_{B} + \Delta + 1\}.$$

Note that we need to add two elements to $index_{u}$ and $index_{v}$ corresponding to the insertion of edge $\{u,v\}$. However, we have added one element each to $index_{u}$ and $index_{v}$ here. The other addition is taken care of either in Case 1 or Case 3.

\paragraph{Case 3:}
Suppose the next pair of edges in the sequence is of type backward followed by a forward edge, say $e_1'e_2$ where $e_1 = \{u_1, v_1\}$ and $e_2 = \{u_2, v_2\}$. In the auxiliary graph $H$, let $u_1, u_2$ belong to node $A$, $v_1$ belongs to node $B$, and $v_2$ belongs to node $C$.

We first update $\Delta' = \Delta - f(u_1) - 2 + L_{B} + 4$, to take care of the fact that the vertices whose indices are to be updated are no more present in the descendants of $A$\footnote{Note that we must have added $f(u_1) - 2$ to $\Delta$ at some point of time earlier in Case 1 to take care of the fact that the vertices present in the descendants of $B$ must be shifted by $f(u_1)+2$.} and the other nodes whose indices are to be modified will get a shift that depends on the length of the E-tour corresponding to $B$.

Now for each vertex $x \in B$ and $\forall z \in index_x$ such that $f(u_1) \le z \le f(u_2)$ we update $z' = z + \Delta'$.  Note that we are modifying only $z \in index_x$ in $f(u_1) \le z \le f(u_2)$, which depends only on $\Delta'$ as the terminals other than $u_1$ and $u_2$ of the E-tour of the node $B$  have either $f$ value less than $f(u_1)$ or  more than $f(u_2)$. For $z \in index_x$ such that $z < f(u_1)$ has been taken care of in either Case 1 or Case 3. For $z \in index_x$ such that $z > f(u_2)$  depends not only on the value $\Delta'$ but also on possibly on the length of some other E-tours, and will be taken care of later either in Case 3 or Case 4.

We also update $index_{u_1}$ and $index_{u_2}$ by:
$$index'_{u_1} = index_{u_1} \cup \{ f(u_2) + \Delta' + 1\}
\qquad
index'_{u_2} = index_{u_2} \cup \{ f(u_2) + \Delta' + 2\}.$$

Other addition of one element each to $index_{u_1}$ corresponding to the insertion of edge $\{u_2,v_2\}$) is taken care in either Case 1 or Case 3. Other addition to $index_{v_1}$  will be taken care of later, either in Case 2 or Case 4, as that depends possibly on the length of some other E-tours.

\paragraph{Case 4:}
Finally, the remaining case is that the next pair of edges in the sequence is both of type backward, say $e_1'e_2'$ where $e_1 = \{u_1, v_1\}$ and $e_2 = \{u_2, v_2\}$. In the auxiliary graph $H$, let $u_2$ belongs to node $A$, $v_2, u_1$ belong to node $B$, and $v_1$ belongs to a node $C$.

Again, we first update $\Delta' = \Delta - f(u_1) - 2 + L_B + 4$ like Case 3. Next, for each vertex $x \in B$ and $\forall z \in index_x$ such that $z \ge f(u_1)$ we update $z' = z + \Delta'$. Note that  we are modifying $\forall z \in index_x$ where $x \in B$ such that $z \geq f(u_1)$, as this depends only on $\Delta$ and not on the length of the other E-tours like in Case 3. This is because there is no terminal node having a higher $f$ value than $f(u_1)$. For $z \in index_x$ such that $z < f(u_1)$ has been taken care of in either Case 1 or Case 3 before.

Finally, we update $index_{u_2}$ and $index_{v_2}$ by:
$$index'_{u_2} = index_{u_2} \cup \{ L_{B} + \Delta' + 2\}
\qquad
index'_{v_2} = index_{v_2} \cup \{ L_{B} + \Delta' + 1\}.$$

Note that we need to add two elements to $index_{u_2}$ and $index_{v_2}$ corresponding to the insertion of edge $\{u_2,v_2\}$. However, we have added one element each to $index_{u_2}$ and $index_{v_2}$ here. The other addition has been taken care of in Case 1 or Case 3 before.

This finishes the description of how the existing E-tours are joined together to get the new E-tour $\eu_T$ corresponding to the updated graph $G'$. Next, we insert the rest of the edges which are all non-tree edges, i.e., both the endpoints of such an edge belong to the same connected component. This also includes the leftover edges $I_{d'}$ from the previous step which became parallel edges in the auxiliary graph $H$. None of these edges affect the E-tour data structure. Recall also that the corresponding sketches have already been updated in the very beginning.

\begin{lemma}
\label{lem:batch-insert}
The algorithm described above can be implemented in $\Oh(1)$ rounds.
\end{lemma}

\begin{proof}
We first argue about the number of rounds. The respective sketches can be updated in $\Oh(1)$ rounds by broadcasting all the insertions. Whether an inserted edge is between distinct connected components in $G$ can be verified in $\Oh(1)$ rounds from the array $C$. From \Cref{claim:ins-forest} we know the auxiliary graph $H$ and a spanning forest $F_H$ in $H$ can be constructed in $\Oh(1)$ rounds. Next, from $F_H$ we first create the sequence $S$ which can be constructed in $\Oh(1)$ rounds from \Cref{lem:seq}. Finally, we update the E-tours. While updating the E-tours for the entire graph is described as a sequential procedure starting from the sequence $S$, it can be implemented in $\Oh(1)$ rounds in \MPC as follows. Since the length of $S$ is only $\Oh(k)$, we can store the entire sequence in the local memory of a single machine. Next, for updating the indices we create $\Oh(k)$ many messages (one for each pair of consecutive edges in $S$) each of size $\Oh(1)$ which we can broadcast to all the machines. From these messages, each machine can update its part of the E-tour stored inside the local memory.
\end{proof}

\subsection{Batch deletions}
\label{subsec:del}

Next, we consider the deletions. In a batch of deletions, for each edge to be deleted there are two cases to consider. If the edge is a non-tree edge (i.e., not part of the maintained spanning forest $F$) then the connectivity information and the spanning forest do not change. However, when a tree edge (which is part of $F$) is deleted either the corresponding component splits into two parts or we might find a replacement (previously non-tree) edge that might connect back these two components. However, notice that we do not have all such replacement edges stored explicitly in the memory. For the deletion of a batch of tree edges, the challenge is to then identify all the potential replacement edges \emph{at once} and build the new spanning forest using these edges.

When a batch of deletions arrives, we first update the sketches as follows.

\paragraph{Updating the sketches.}
We denote the set of deleted edges as $D \subseteq U$. We update the sketches independently for each deleted edge similar to \Cref{sec:pure-streaming}. In particular, for each edge $e=\{u,v\}$ such that $e \in D$ we update the sketches $S_u$ and $S_v$. As before, $S_u$ and $S_v$ can be updated in $\Oh(1)$ rounds by broadcasting the information that the sketches of $S_u$ and $S_v$ need to be updated with deletion of edge $e = \{u,v\}$ for each $e \in U_D$. {However, unlike in \Cref{sec:pure-streaming}, where we maintain only one sketch per vertex that performs as desired with high probability, here we maintain $t = \Oh(\log n)$ independent sketches for each vertex $v \in V(G)$. Each of these sketches operates with a constant success probability, i.e., each sketch requires $\Oh(\log^2 n)$ bits of space (see \Cref{lem:sketch1}).} Let us denote them by $\mathcal{S}_v = \{S^{(1)}_v, S^{(2)}_v, \ldots, S^{(t)}_v \}$.

Our main algorithm proceeds as follows. We first remove the non-tree edges from $G$. We only need to update the sketches corresponding to these edges, they do not affect the maintained spanning forest and so E-tours. Next, we consider the tree edges. For simplicity, we consider the deletions only in a single connected component of $G$. Deletions across all components can be handled in parallel in a similar fashion. Removing tree edges splits the corresponding spanning tree into several subtrees $T_1, T_2, \dots, T_p$ where $p \le 2k$. Let $Z_i$ be the set of vertices in the component corresponding to the tree $T_i$. In parallel for each $i \in [p]$, we merge the sketches of the vertices in $Z_i$ to get the sketches $\mathcal{S}_{Z_i}$. Similar to the single edge deletion case, for each $i$ this can be done by first broadcasting the ID of the E-tour $\eu_{T_i}$ and then merging the sketches of the vertices in $Z_i$.

Next, similar to the insertion case, we construct an auxiliary graph $H$ as follows. Each $v \in V(H)$ corresponds to a connected component of the graph $G'$ after the tree edges have been removed. Note that each such component contains at least one endpoint of a tree edge $e = \{u,v\}$.

To update the E-tour trees, which now splits into several parts after the deletions, we employ an inverse procedure of the methods that we use to update the Euler tours for processing insertions.  We describe the idea very briefly as the procedure is entirely symmetric to the batch insertion scenario. In particular, consider the graph $H$ with the deleted tree edges as $E(H)$. We create the same auxiliary sequence $S$ from $H$ which is of length $\Oh(k)$. We look at each consecutive pair of edges in $S$ which again give rise to four different cases. However, in each case, while we generate the messages in a very similar way as before the messages themselves will be different in the following way: the shift-index operations will come with different signs (as the indices will decrease after deletions) and the update-index operations will remove indices from the sets $\ind_v$ corresponding to terminal vertices $v$.

Our next goal is to construct a spanning forest of the graph $H$ with vertices $V(H)$ and without any edges (i.e., after removing the deleted edges) using the maintained sketches $\mathcal{S_v}$ for each $v \in V(H)$.

\paragraph{Constructing $F_H$.}
We find a spanning forest $F_H$ in $H$ in $t = \Oh(\log k)$ iterative steps. The algorithm follows the approach of Ahn \etal \cite{AGM12} output a spanning forest (also described in \Cref{sec:pure-streaming}). While in \cite{AGM12} the algorithm is applied on the entire graph $G$ at the end of the stream, here we apply it only on the auxiliary graph $H$ but after every batch of updates.


In the first step, we query sketch $S^{(1)}_v$ for a replacement edge for each vertex $v$ from (if exists). Then we merge the vertices present in the same connected component in the graph $H_1(V_1, E_1)$ where $V_1=V(H)$ and $E_1$ is the set of replacement edges found in the current step. Similarly, in the $i$-th round with $2 \leq i \leq t$, $V_i$ is the set of supernodes where each supernode corresponds to a connected component in $H_{i-1}$. We use the $i$-th sketches of the corresponding vertices to find possible replacement edges from each supernode in $V_{i}$ (if exists). Let $E_i$ be the set of edges between the nodes in $V_i$ that were found from the sketches. The algorithm terminates after $t=\Oh(\log k)$ rounds when each node in $V_t$ corresponds to the connected components of $H$. The correctness of the above procedure follows from the fact that in each step the expected size of the graph decreases at least by a factor of $2$ using the linearity of expectation. Hence, the procedure terminates in $\Oh(\log k)$ steps. We now show how to implement the above procedure efficiently in \MPC.

\begin{lemma}
\label{lem:del-forest}
The spanning forest $F_H$ can be constructed in $\Oh(1/\spacexp)$ rounds.
\end{lemma}

\begin{proof}
First, observe that the set of vertices $V(H)$ can be found in $\Oh(1/\spacexp)$ rounds as follows. For a vertex $v \in V(H)$ we first broadcast the id of the E-tour $\eu_{T_v}$ and then merge the sketches of the vertices in $Z_v$. The broadcasting can be done in $\Oh(1)$ rounds and the merging step can be done in $\Oh(1/\spacexp)$ rounds. Note that we can do the above step for all the nodes in $H$ and all $\Oh(\log n)$ independent sketches in parallel. Observe that since $|V(H)| \le 2k$, for $k = \Oh(n^\spacexp)$ we can store $V(H)$ in the local memory of a single machine. However, for each $i \in [k]$ the total space need to store the sketches corresponding to each vertex in $H$ is $\Oh({\log ^3 n})$. Hence in a single machine, we can gather all the information to simulate the algorithm described above, for constructing $F_H$ locally.
%
\end{proof}

Finally, we need to insert back the edges of the spanning forest $F_H$ to the intermediate spanning forest of $G$ (after the deletions) to find the new spanning forest $F'$. The corresponding Euler tours due to these insertions can be updated using the same algorithm for processing batch insertions from \Cref{subsec:ins}.

\begin{lemma}
\label{lem:batch-delete}
The algorithm described above can be implemented in $\Oh(1/\spacexp)$ rounds.
\end{lemma}

\begin{proof}
We first argue about the number of rounds. The respective sketches can be updated in $\Oh(1)$ rounds by broadcasting all the deletions. Whether a deleted edge is a non-tree can be verified in $\Oh(1)$ rounds from the array $C$. After the deletions, we can update the E-tour trees following a similar procedure as in the insertion case that can also be done in $\Oh(1)$ rounds by \Cref{lem:batch-insert}. Next, from \Cref{lem:del-forest}, we can construct a spanning forest $F_H$ in the auxiliary graph $H$ in $\Oh(1/\spacexp)$ rounds. Once we find $F_H$, which is of size $\Oh(k)$, we insert back the edges in $F_H$ in our graph. This amounts to updating the E-tours under a batch of insertions and can be dealt with as before in $\Oh(1)$ rounds using \Cref{lem:batch-insert}.
\end{proof}

Hence, from the above description, we have the following result:

\begin{theorem}
\label{thm:con-batch}
Let $0 < \spacexp < 1$ be an arbitrary constant. Given an undirected graph $G$ with $n$ vertices, we can maintain the connectivity of $G$ to process a batch of $\Oh(n^{\spacexp}/{\log^3 n})$ updates in $\Oh(1/\spacexp)$ rounds on an \MPC with sublinear local memory $\lspace=\Oh(n^{\spacexp})$ and $\Oh(n {\ \log^3 n})$ total memory. Furthermore, within the same bounds, the algorithm can maintain a spanning forest of $G$. Moreover, it is assumed that the total length of the update stream is a polynomial in $n$.
\end{theorem}

\begin{proof}
There are three parts to the proof: join and split operation of multiple E-tour trees, maintaining the spanning forest, and the round complexity analysis.

For the join and split operation of multiple E-tour trees, note that the join operation is used for insertion and the split operation is used for deletion. First, consider the insertion case. We only need to argue that the messages received by each machine correctly update the final E-tour as the rest is clear from the description. Towards this note that for a pair of consecutive edges $e_ie_j$ in the sequence $S$, one of the four cases can arise based on whether $e_i$ or $e_j$ is a forward edge or not. In each such case, we present a constructive algorithm and explicitly describe the messages that we broadcast, from which the correctness can be verified readily. In the deletion case, the split of an E-tour tree into multiple E-tour trees, the implementation details are similar to the join operation in the insertion case and their correctness can be argued similarly.

For maintaining the spanning forest, recall that we have argued the correctness of our streaming algorithm in \Cref{sec:pure-streaming}. Then we have shown its \MPC implementation for a single update in \Cref{sec:connect-single}. The fact that the algorithm for batch updates maintains the spanning forest correctly w.h.p.\ follows from the fact that we are essentially implementing the streaming algorithm in \Cref{sec:pure-streaming}. However, a straightforward implementation of the algorithm in \Cref{sec:connect-single} requires a number of rounds proportional to the number of updates. That our implementation maintains a spanning forest correctly follows from (i) the correct implementation of the join and split operations of multiple E-tour trees in parallel, and (ii) we use $\Oh(\log n)$ independent sketches here similar to \cite{AGM12} since we build on their approach for finding the spanning forest of $F_H$ in \MPC.

The round complexity of our algorithm follows from the fact that a batch of insertions is processed in $\Oh(1)$ rounds (\Cref{lem:batch-insert}) and a batch of deletions in $\Oh(1/\spacexp)$ rounds (\Cref{lem:batch-delete}), and by considering the insertions and deletions in a single batch of updates in two consecutive steps.
\end{proof} 
\section{Applications of \textsc{Connectivity}}
\label{sec:mst}

In this section, extending the connectivity algorithm in \Cref{sec:connect-batch} we present an algorithm in the \MPC model under batch updates for (i) an exact minimum spanning forest (MSF) algorithm in insertion-only streams, (ii) approximating the weight of MSF and (iii) testing whether the input graph is bipartite, both in dynamic streams.

The results on the minimum spanning forest are formally stated in the following theorem.

\begin{theorem}
\label{thm:mst}
Let $0 < \spacexp < 1$ be an arbitrary constant. Given an undirected graph with $n$ vertices, on an \MPC with $\lspace=\Oh(n^{\spacexp})$ local memory and $\widetilde{\Oh}(n)$ total memory we can process a batch of $\Oh(n^{\spacexp}/{\log ^3 n})$ updates and maintain in $\Oh(1/\spacexp)$ rounds {(i) an exact minimum spanning forest for insertion-only updates, and (ii) a $(1+\epsilon)$-approximate minimum spanning forest for arbitrary updates}. Moreover, it is assumed that the total length of the update stream is a polynomial in $n$.
\end{theorem}

The algorithm in insertion-only streams is presented in \Cref{sec:mst-xat} where we discuss our result in dynamic streams (that maintains an approximation to the weight of the MSF) in \Cref{sec:mst-apx}. The result on bipartiteness is discussed in \Cref{sec:bip}.

\subsection{Exact MSF in insertion-only streams}
\label{sec:mst-xat}

The algorithm maintains a current minimum spanning forest $F$ at any point in time. When an edge $e=\{u,v\}$ arrives, we check whether $u$ and $v$ are in the same component or not. If no, then we add $e$ to the current $F$. Otherwise, we find the heaviest edge $e'$ in the path from $u$ to $v$ in $F$. If the weight of $e'$ is more than $e$, we delete $e$ from $F$ and add $e'$ to $F$. Otherwise, there is no change to $F$. The space complexity is clearly $\tOh(n)$. Note that this algorithm is a folklore algorithm whose processing time per update can be $\Omega(n)$ in the worst case. Our contribution here is to show that we can implement the algorithm in the \MPC model such that we can process a batch of $\Oh(n^\phi)$ updates that can be done in $\Oh(1)$ rounds.

\subsubsection{\MPC implementation for a single update}
\label{sec:mst-insertion-only-single-step}

Recall our connectivity algorithm in \Cref{sec:connect-single}. We maintain the same data structure along with E-tours for each tree in the minimum spanning forest (we are maintaining). Recall that we have discussed the implementations of three operations in \MPC: rooting, joining, and splitting. Here, we introduce another operation \emph{Identify-Path} (defined below) which also can be implemented by communicating only $\Oh(1)$ size information.

\paragraph{Identify-Path.} We are given an E-tour $\eu_T$ of a tree $T$ and two vertices $u$ and $v$ in $T$. The objective is to report all the edges in the path between $u$ and $v$ in $T$.

\begin{lemma}
   Consider the operation Identify-Path defined above. This can be implemented on an E-tour in $\Oh(1)$ rounds in \MPC.
\end{lemma}
\begin{proof}
    Assume that $f(u) < f(v)$. The case when $f(u)>f(v)$ is analogous. Any edge $\{a,b\}$ in the unique path from $u$ to $v$ in $T$ satisfy one of the following properties depending on whether $\ell(u) <\ell(v)$ or not:

\begin{description}
     \item[(i) $\ell(u)<\ell(v)$:] One of the following must be true.
     \begin{itemize}
         \item $f(a), f(b) \leq f(u)$; $\ell(a), \ell(b)\geq \ell(u)$; and $\ell(a), \ell(b) \leq \ell(v)$.
         \item  $\ell(a), \ell(b)\geq \ell(u)$; $f(a), f(b) \geq f(v)$; and $f(a), f(b) \geq f(v)$.
     \end{itemize}
     \item[(ii) $\ell (u) > \ell (v)$:] $f(a),f(b)>f(u)$; $f(a),f(b)<f(v)$;  $\ell(a),\ell(b) < \ell(u)$; $\ell(a), \ell(b) > \ell (v)$.
\end{description}

We can broadcast the value of $f(u),f(v),\ell(u)$, and $\ell(v)$ to each machine in $\Oh(1)$ rounds. Then each edge $\{a,b\}$ can decide whether it is in the path from $u$ to $v$ by checking (i) or (ii) depending on $\ell(u) < \ell(v)$ or not, respectively.


Now, we are ready to discuss how we update (in $\Oh(1)$ rounds) the minimum spanning forest $F$ and the E-tours when an edge $e=\{u,v\}$ is inserted.

\paragraph{$u$ and $v$ are in different connected components.}
Here, we want to merge the two components one having $u$ and the other one having $v$. Here the update to the data structure is exactly the same as the Insertion operation that we discussed in \Cref{lem:analysis-single}. Hence, this can be done in $\Oh(1)$ rounds.

\paragraph{$u$ and $v$ are in the same connected components.}
Here we perform the Identify-Path operation such that machines can detect all the edges on them that are in the path between $u$ and $v$ in $F$. Then, we can identify the edge $e'$ with the maximum weight in the path from $u$ to $v$ in $F$ in $\Oh(1/\phi)$ rounds using a broadcast tree argument or sorting. We are done when the weight of $e'$ is less than that of $e$. Otherwise, we delete edge $e$ from $F$ and then finally we insert edge $e'$. The corresponding update to the data structures and E-tours are exactly the same as the Delete and Insert operation in \Cref{lem:analysis-single}, and hence can be done in $\Oh(1)$ rounds.
\end{proof}

\subsubsection{\MPC implementation for batch updates}
\label{sec:mst-insertion-only-batch}

We divide the analysis into two parts: (1) for each $\{u,v\}$ in $I$, $u$ and $v$ are in different components or (2) for each $\{u,v\}$ in $I$, $u$ and $v$ are in the same component.

\paragraph{Case 1: (edges in $I$ are not in the same component).}  Consider the components $\mathcal{C}$ of $G$ between which the edges in $I$ are present, which are at most $O(n^\phi)$ many. Consider a subset $X$ of $I$ as follows. For any two components, $C_1$ and $C_2$ in $\mathcal{C}$ such that there is at least one edge having endpoints in both $C_1$ and $C_2$, $X$ has exactly one such edge with the minimum weight. Then we use our connectivity algorithm to insert the batch of edges in $X$.

\paragraph{Case 2: (edges in $I$ are in different components).}  We first find a set $I'$ and delete them from the graph $G$. For each $e=\{u,v\}$ in $I$, in parallel, we perform the Identify-Path operation such that machines can detect all the edges on them that are in the path between $u$ and $v$ in $F$. This can be done in $\Oh(1)$ rounds as we need to broadcast $f(u), f(v), \ell(u)$, and $\ell(v)$ for each $\{u,v\}$ in $I$. Then, for each edge $e=\{u,v\}$, we can identify the edge $e'$ with the maximum weight in the path from $u$ to $v$ in $F$ in $\Oh(1/\phi)$ rounds. Let $I'$ be the set of such edges $e'$. As already pointed out, we perform deletion of the edges in $I'$. This is equivalent to deleting a batch of $\Oh(n^{\spacexp})$ edges to maintain a spanning forest. Hence, by \Cref{lem:batch-delete}, all the edges in $I'$ can be deleted and the data structure along with E-tours can be updated in $\Oh(1/\spacexp)$ rounds. Now, finally, the objective is to insert the edges in $I \cup I'$. Note that none of the edges in $I \cup I'$ are in the current minimum spanning forest (as we have deleted the edges in $I'$). So, the insertion of the edge in $I \cup I'$ can be done in the same way as we have handled Case 1: we find a suitable subset of edges $X \subseteq I\cup I'$ and insert them. Note that  all steps discussed in Case 2 can be performed in $\Oh(1)$ rounds. 

\subsection{Approximate minimum spanning forest}
\label{sec:mst-apx}

In this section, we present an algorithm for maintaining a $(1+\epsilon)$-approximate minimum spanning forest in dynamic streams. To begin with, we present a simpler algorithm that only maintains the approximate weight of the minimum spanning forest.

\subsubsection{Approximate weight of an MSF}
\label{sec:mst-apx1}

We now describe an algorithm to maintain a $(1+\epsilon)$-approximation to the weight of the minimum spanning forest in dynamic streams. We reduce our problem to maintaining connectivity which is an adaptation of the idea of Chazelle et. al \cite{CRT05}.


Let the given graph be $G=(V, E)$ with edge weights in the range $[1, W]$ where $W$ is bounded by $\poly(n)$. Wlog, we assume that $G$ is connected otherwise we can apply the same algorithm on each connected component of $G$ in parallel. Consider $t+1$ many graphs $G_0, G_1, \dots, G_t$ where $t= \lceil \log_{1+\epsilon} W \rceil$ and each $G_i$ is a subgraph of $G$ consisting of the entire vertex set $V$ but only the edges of weight at most $w_i = (1+\epsilon)^i$ from $E$. For $0 \le i \le t$, let $cc(G_i)$ be the number of connected components of $G_i$. To find the weight of the approximate minimum spanning tree $T$ of $G$, we consider the difference between $cc(G_{i+1})$ and $cc(G_{i})$ for $0 \le i \le t$. In particular, we have
\begin{align}
\label{eq:mst}
    w(T) &\le n - (1+\epsilon)^t + \sum_{i=0}^t \lambda_i cc(G_i) \le (1+\epsilon)w(T) \enspace.
\end{align}
where $\lambda_i = (1+\epsilon)^{i+1} - (1+\epsilon)^i$. See \cite[Lemma 3.4]{AGM12} for a proof.

Our algorithm for maintaining an $(1+\epsilon)$ approximation to the weight $w(T)$ proceeds as follows. In the preprocessing phase, we construct the graphs $G_0, G_1, \dots, G_t$ from $G$. In each $G_i$, we maintain the number of connected components $cc(G_i)$ using the algorithm in \Cref{sec:connect-batch} under batch updates. This can be done by counting the distinct $C[i]$ values which takes $\Oh(1)$ rounds. Whenever a query arrives we compute the quantity in \Cref{eq:mst} in the local memory of a single machine and output it. The correctness and round complexity follow from the above description and the guarantees of \Cref{thm:con-batch}. 
\subsubsection{Finding an approximate MSF}
\label{mst-exact-search}

In this section, we show that we can extend algorithm in the the previous section to even output a $(1+\eps)$-approximate minimum spanning forest. As before, let the given graph be $G=(V, E)$ with edge weights in the range $[1, W]$ where $W$ is bounded by $\poly(n)$. We again assume that $G$ is connected, otherwise, we can apply the same algorithm on each connected component of $G$ in parallel. Consider $t+1$ many graphs $G_0, G_1, \dots, G_t$ where $t= \lceil \log_{1+\epsilon} W \rceil$ as before i.e., each $G_i$ is a subgraph of $G$ with vertex set $V$ and edges of weight at most $w_i = (1+\epsilon)^i$ from $E$. For $1 \le i \le t$, let $F_i$ be the spanning forest of $G_i$.

Our algorithm for maintaining an $(1+\epsilon)$ approximate MSF proceeds as follows. In the preprocessing phase, we construct the graphs $G_0, G_1, \dots, G_t$ from $G$. In each $G_i$, we maintain a spanning forest $F_i$ using the algorithm in \Cref{sec:connect-batch} under batch updates. Recall that in our connectivity algorithms (see \Cref{sec:pure-streaming}) we maintain a component id $C[v]$ for each vertex $v \in V(G)$ that stores the id of the component in which the vertex $v$ lies. Hence for each $G_i$, for $0 \le i \le t$, we maintain here a component id vector $C_i$.

To construct $F$ for the entire graph $G$, for each $F_i$ we consider each edge $e=\{u,v\}$ in $F_i$ in parallel and check if the component ids of $u$ and $v$ in $C_{i-1}$ are the same. We add $e$ to $F$ if and only if $C_{i-1}[u]$ and $C_{i-1}[v]$ are different. This can be done in $\Oh(1)$ rounds for each edge $e$ and for each graph $G_i$. So all together our update algorithm takes $\Oh(1)$ rounds from the guarantees of \Cref{thm:con-batch}.

\paragraph{Correctness of our algorithm:}
Now we proceed to prove the correctness. First notice that the vertices $u$ and $v$ have different $C_{i-1}$ ids iff they are disconnected in $G_{i-1}$. We add this edge $\{u,v\}$ to $F$ from $F_i$, and our conditions are uniquely met for this particular $i$. For $j \ge i$ we have $C_j [u] = C_j[v]$, so we do not add this edge or any other edge between two vertices in that component. For $j < i$, there is no path between $u$ and $v$ in $F_j$.

Also, notice that $F_i$ is not necessarily a subset of $F_j$, for $i \le j$, due to the nature of our maintenance algorithm. This is because the maintained sketches might return a different edge while it is queried for an edge on some vertex $v$ in different $G_i$s. So we may have different edges in different $F_j$, for $j \ge i$, connecting two distinct components of $G_{i-1}$ (and so in $F_{i-1}$). However, while the edges in the spanning forests for each $F_i$ might vary, their component structure is still the same i.e., two vertices in $F_i$ are connected iff they are connected in $G_i$. Moreover, they are also connected in $F_j$ for each $j \ge i$, albeit possibly through different paths.

Hence, when we add an edge $\{u,v\}$ to $F$ from some $F_i$ we know that there was no path connecting $u$ to $v$ in any of $G_0, \ldots, G_{i-1}$. So $F$ is indeed a forest. The forest $F$ is also a spanning one. For contradiction, suppose not, and towards this there is an edge between some $u$ and $v$ but they are in different components in $F$. But from the guarantees of the connectivity algorithm we know for some $F_i$ they must be in the same component, suppose wlog is connected by the edge $\{u,v\}$. Let $F_{i^*}$ be the first such $F_i$. Then our algorithm would add this edge while considering the edges of $F_{i^*}$, which is a contradiction.

The fact that $F$ is indeed a $(1+\eps)$-approximate minimum spanning forest follows from the observations in \Cref{sec:mst-apx1} with the fact that the number of edges added to $F$ from some $G_i$ remains the same as before.

\subsection{Bipartiteness}
\label{sec:bip}

In this section, we provide an algorithm that tests whether the input graph is bipartite or not in dynamic streams. The result is formally stated in the following theorem.

\begin{theorem}
\label{thm:bipartite}
Let $0 < \spacexp < 1$ be an arbitrary constant. Given an undirected graph $G$ with $n$ vertices, we can maintain the bipartiteness of $G$ to process a batch of $\Oh(n^{\spacexp}/{\log ^3 n})$ updates in $\Oh(1/\spacexp)$ rounds on an \MPC with sublinear local memory $\lspace=\Oh(n^{\spacexp})$ and $\widetilde{\Oh}(n)$ total memory. Moreover, it is assumed that the total length of the update stream is a polynomial in $n$.
\end{theorem}

Consider the following graph $G' = (V', E')$. For each vertex $v \in V$ create two vertices $v_1, v_2 \in V'$ and for each edge $e= \{u, v\} \in E$, create two new edges $\{u_1, v_2\}$ and $\{u_2, v_1\}$ in $E'$. Then, from \cite[Lemma 3.3]{AGM12}, we know that $G$ is bipartite if and only if the number of connected components in $G'$ is exactly twice that of $G$.

\begin{lemma}[\cite{AGM12}]
Let $K$ be the number of connected components in $G$. Then $G'$ has $2K$ connected components if and only if $G'$ is bipartite.
\end{lemma}

Our \MPC algorithm for maintaining the bipartiteness of $G$ proceeds as follows. In the preprocessing phase, we construct the graph $G' = (V', E')$ from $G$ as described above. Our \MPC algorithm in \Cref{sec:connect-batch} can maintain the number of connected components of the given graph from the distinct $C[i]$ values. We run this algorithm on both $G$ and $G'$. The correctness and round complexity follow from the above description and the guarantees of \Cref{thm:con-batch}. Furthermore, a single insertion or deletion in $G'$ inserts or deletes only two edges in $G'$, respectively.  
\section{Approximate maximum matching}
\label{sec:matching}

In this section, we present \MPC algorithms for the approximate maximum matching problem in both insertion-only and dynamic streams. We describe the algorithms for finding an approximate matching in \Cref{sec:mat-find} and estimating the size of the maximum matching in \Cref{sec:mat-est}.

\subsection{Finding an approximate matching}
\label{sec:mat-find}

In this section, we prove our results on finding an approximate matching in insertion-only streams and dynamic streams (insertion-deletions streams) in \Cref{thm:match-find} and \Cref{thm:match-find1}, respectively.

\begin{theorem}
\label{thm:match-find}
Let $0 < \spacexp < 1$ be an arbitrary constant and let $0 < \kappa < \spacexp$ and $\alpha > 1$ be arbitrary. Given an undirected graph $G$ with $n$ vertices, on an \MPC with $\lspace = \Oh(n^{\spacexp})$ local memory we can process a batch of $\Oh(\lspace)$ updates and maintain in $\Oh(1)$ rounds an $\Oh(\alpha)$-approximate maximum matching in $G$ with $\widetilde{\Oh}(n/\alpha)$ total memory for insertion-only updates.

\end{theorem}

\begin{proof}
Our algorithm maintains a matching $M$ which is either a maximal matching or a matching of size at most $cn/\alpha$ (in the graph seen so far), where $c$ is a suitable constant. On each machine, we have a set of edges stored and also the information about the edges that are present in $M$. Now, let us discuss how to update the information over the machines when a batch $I$ of $\Oh(n^{\spacexp})$ insertions arrive. If $|M|\geq cn/\alpha$, we do not update anything. Otherwise, we proceed as follows.

We broadcast $I$ to all machines and listen from the machines about the edges (in $I$) whose endpoints coincide with the endpoints of some edges in $M$. Note that this step can be performed in $\Oh(1)$ rounds. Let $I' \subseteq I$ be the set of edges none of whose endpoints are present in $M$. We greedily add the edges in $I'$ to the current maximal matching till the size of the matching exceeds $cn/\alpha$. When we are asked to report a matching, we output the currently stored matching $M$ in the memory. Observe that the algorithm uses $\tOh(n/\alpha)$ space in total. The correctness and the round complexity of the algorithm (for each update and query) follow from the description.
\end{proof}

We now describe our algorithm for dynamic streams. The main result is as follows.

\begin{theorem}
\label{thm:match-find1}
Let $0 < \spacexp < 1$ be an arbitrary constant and let $0 < \kappa < \spacexp$ and $\alpha > 1$ be arbitrary. Given an undirected graph $G$ with $n$ vertices, on an \MPC with $\lspace = \Oh(n^{\spacexp})$ local memory we can process a batch of $\Oh(\lspace^{1 - \kappa})$ updates and maintain in $\Oh(\log(1/\kappa))$ rounds an $\Oh(\alpha)$-approximate maximum matching in $G$ with $\widetilde{\Oh}(\max\{n^2/\alpha^3,n/\alpha\})$ total memory for arbitrary updates. Moreover, it is assumed that the total length of the update stream is a polynomial in $n$.

\end{theorem}

Our approach relies on extending the known algorithms for finding an approximate matching in dynamic streams and implementing those in the MPC model suitably. However, it is not the case that all (sketching-based) algorithms in the streaming literature (to find an approximate matching) can be extended to \MPC. In particular, there are three different papers~\cite{AKLY16, CCEHMMV16,Konrad15} that give essentially the same result. However, our techniques rely on the streaming algorithm of Assadi, Li, Khana, and Yaroslavtsev \cite{AKLY16}. We do not see how the algorithms from the other two papers~\cite{CCEHMMV16,Konrad15} can be extended to the \MPC setting such that efficient update complexity can be achieved.

\subsubsection*{Overview of the streaming algorithm of Assadi-Li-Khana-Yaroslavtsev \cite{AKLY16}}

The algorithm in \cite{AKLY16} assumes that we know $\opt'$ which is a $2$-factor approximation on the size of the maximum matching. We can run $\Theta(\log n)$ instances of the algorithm for $\Theta(\log n)$ different guesses for $\opt '$, i.e., $\opt'=n/2$, $\opt'=n/4$, and so on. Then finally we can report the maximum size of the matching found in any of the $\Theta(\log n)$ instances. Without loss of generality, assume that the input graph is bipartite and let $L \sqcup R$ be the bipartition of the vertex set. Otherwise, we can randomly partition the vertex set $V$ into two parts, $L$ and $R$, by using a hash function chosen randomly from a pairwise independent hash family. If we find the maximum matching restricted to the edge set between $L$ and $R$, then one can argue that w.h.p. the size of this matching is within a constant factor of the size of the maximum matching in the original graph $G$.

\paragraph{Pre-processing.}
Let $\beta=\lceil \opt'/\alpha \rceil$ and $\gamma=\lceil \opt'/\alpha^2 \rceil$. The algorithm randomly partitions the vertex sets $L$ and $R$ into $\alpha$ groups by using hash functions $h_L: L \rightarrow [\beta]$ and $h_R: R \rightarrow [\beta]$ chosen independently from a pairwise independent hash family. Let $L_i=\{v \in L | h_L(v)=i\}$ and $R_i=\{v \in R | h_R(v)=i\}$, where $i \in [\beta]$. For each $L_i$, the algorithm assigns $L_i$ with $\gamma$ number of $R_j$'s independently and uniformly  at random and with replacement. If $R_j$ is assigned to $L_i$, then $(L_i,R_j)$ is said to be an \emph{active pair}. Note that the number of active pairs is ${\tOh}( \max\{n^2/\alpha^3, n/ \alpha\})$. For each active pair $(L_i,R_j)$, an $\ell_0$-sampler is initiated, i.e., a sketch for the edge set $E(L_i,R_j)$ of size $\Oh(\log ^3 n)$ as stated in \Cref{lem:sketch-A-B}. This implies the space complexity of the algorithm of \cite{AKLY16} is $\tOh \left(\max\{n^2/\alpha^3, n /\alpha\}\right)$.

\paragraph{Streaming Phase:} When an edge $e=\{u,v\}$ is inserted or deleted, we first determine the groups $L_i$ and $R_j$ such that $u \in L_i$ and $v \in R_j$. If $(L_i,R_j)$ forms an active pair, then update the sketch for $E(L_i,R_j)$ accordingly.

\paragraph{Post-processing.} Let $H$ be the subgraph formed by the set of edges we get from the sketches of the active pairs. Then \cite{AKLY16} shows the following:

\begin{lemma}
\label{lem:mat-H}
The maximum matching of $H$ is an $\Oh(\alpha)$ approximation to the maximum matching of $G$ w.h.p.
\end{lemma}


Apart from the streaming algorithm by \cite{AKLY16}, our algorithm of \Cref{thm:match-find1} relies on the algorithm in the \MPC model \cite{NO21} that maintains a $2$-approximate maximum matching in a black box manner. The result of \cite{NO21} is formally stated as follows.
\begin{proposition}
\label{prop:onak}
Given an undirected graph $G$ with $n$ vertices we can find an $2$-approximate maximum matching in $G$ in \MPC with $S = \Oh(n^\spacexp)$ local memory in $\Oh(\log 1/\kappa)$ rounds to process a batch of update of size $\Oh(\lspace^{1-\kappa})$ with ${\widetilde{\Oh}(m)}$ global memory. Moreover, the algorithm spends $\Oh(\log(1/\kappa))$ rounds to report a query. For constant $\kappa$ this gives an $\Oh(1)$ round algorithm and for $\kappa = \Oh(1/\log n)$ this gives an $\Oh(\log \log n)$ round algorithm.
\end{proposition}

Now we are ready to prove \Cref{thm:match-find1}. The high-level idea is to maintain the information same as that of \cite{AKLY16} and run the \MPC algorithm by \cite{NO21} on the graph $H$ that can be formed by considering the edges we get from the sketches corresponding to the active pairs. 
\begin{proof}[Proof of \Cref{thm:match-find1}]
As \cite{AKLY16}, we perform all prepossessing steps before the start of the stream and the required information like $h_L,h_R$ and the set of active pairs are stored over the machines in a distributed fashion. Note that $h_L$ and $h_R$ can be stored in all the machines since $\Oh(\log n)$ bits are enough to store them. The sketches/$\ell_0$-samplers (each of size $\Oh(\log ^3 n)$) corresponding to the active pairs are also stored over the machines in a distributed fashion. Our algorithm also stores the outcome of each $\ell_0$-sampler at any instance of time along with the sketches. Let $U$ be the set of updates under consideration. For an edge $e=\{u,v\}$ in $U$ (either we want to insert or delete $e$), we say $e$ is active update if $\left(L_{h_L(u)},R_{h_R(v)}\right)$ is an active pair. This can be verified in $\Oh(1)$ rounds first by broadcasting the set $U$ and then receiving the information about which edges in $U$ are active updates.
Let $U'\subseteq U$ be the active updates.

Let $X$ be the outcome of the $\ell_0$-samplers corresponding to the active pairs $\left(L_{h_L(u)},R_{h_R(v)}\right)$ where $\{u,v\} \in U'$. Note that $\size{X}\leq \size{U}=\Oh(\lspace^{1-\kappa})$. Now, we gather $X$ in $\Oh(1)$ rounds and delete $X$ from the graph $H$ in $\Oh(\log 1/\kappa)$ rounds by using the algorithm by \cite{NO21} (see \Cref{prop:onak}). For each edge $e=\{u,v\}$ in $U'$ we update the sketches corresponding to the active pair $\left(L_{h_L(u)}, R_{h_R(v)}\right)$ accordingly. Observe that this can be done by just broadcasting $U'$. Let $Y$ be the outcome of the $\ell_0$-samplers corresponding to the active pairs $(h_L(u),h_R(v))$ where $\{u,v\} \in U'$. Note that $\size{Y}\leq \size{U}=\Oh(\lspace^{1-\kappa})$. Now, we gather $Y$ in $\Oh(1)$ rounds and insert $Y$ to the graph $H$ in $\Oh(\log 1/\kappa)$ rounds again using the algorithm by \cite{NO21} (see \Cref{prop:onak}). Here, one crucial point is that the input to the \MPC algorithm of \Cref{prop:onak} is dependent on the randomness of the sketches/$\ell_0$-samplers. This is not a problem as the random bits of the \MPC algorithm and that of the sketches are independent. When we get a query to report an $\Oh(\alpha)$-approximate matching, we use the output of the \MPC algorithm of \Cref{prop:onak} to report a $2$-approximate matching in $H$. Note that, due to \Cref{lem:mat-H}, this indeed produces an $\Oh(\alpha)$-approximate matching.

From the above description, the number of rounds to update a batch of size $\Oh(\lspace^{1-\kappa})$ is $\Oh(\log 1/\kappa)$. As the dynamic streaming algorithm of \cite{AKLY16}  produces a sparse graph $H$ with number of edges $\tOh(\max\{n^2/\alpha^3, n/\alpha\})$ and we are running \MPC algorithm on graph $H$ (that uses total memory of $\tOh(|E(H)|)$), the total memory used by our algorithm is $\tOh(\max\{n^2/\alpha^3, n/\alpha\})$.
\end{proof}

\subsection{Estimating the size of a matching}
\label{sec:mat-est}

In this section, we prove our results on estimating the size of the maximum matching in insertion-only streams and dynamic streams in \Cref{thm:match-size} and \Cref{thm:match-size1}, respectively.

\begin{theorem}
\label{thm:match-size}
Let $0 < \spacexp < 1$ be an arbitrary constant and let $0 < \kappa < \spacexp$ and $ \alpha \leq \sqrt{n}$ be arbitrary. Given an undirected graph $G$ with $n$ vertices, on an \MPC with $S = \Oh(n^{\spacexp})$ local memory we can process a batch of $\Oh(S)$ updates and maintain in $\Oh(1)$ rounds an $\Oh(\alpha)$-approximate to the size of maximum matching in $G$ with $\widetilde{\Oh}(n/\alpha^2)$ total memory for insertion-only updates. Moreover, it is assumed that the total length of the update stream is polynomial in $n$.
\end{theorem}

\begin{theorem}
\label{thm:match-size1}
Let $0 < \spacexp < 1$ be an arbitrary constant and let $0 < \kappa < \spacexp$ and $\alpha \leq \sqrt{n}$ be arbitrary. Given an undirected graph $G$ with $n$ vertices, on an \MPC with $S = \Oh(n^{\spacexp})$ local memory we can process a batch of $\Oh(S^{1 - \kappa})$ updates and maintain in $\Oh(\log(1/\kappa))$ rounds an $\Oh(\alpha)$-approximation to the size of maximum matching in $G$ with $\widetilde{\Oh}(n^2/\alpha^4)$ total memory for arbitrary updates. Moreover, it is assumed that the total length of the update stream is polynomial in $n$.
\end{theorem}

Both the algorithms in this section rely on the streaming algorithms by Assadi, Khanna, and Li~\cite{AKL21} that can report an $\Oh(\alpha)$ approximation to the size of the maximum matching; that uses $\widetilde{O}(n/\alpha)$ space and $\widetilde{O}(n/\alpha^2)$ space in insertion-only streams and in dynamic streams, respectively.

\subsubsection*{Overview of the streaming algorithms by Assadi-Khanna-Li~\cite{AKL21}}

The streaming algorithms of \cite{AKL21} for insertion-only and dynamic streams are based on a meta-algorithm called $\mbox{\textsc{Tester}}(G,k)$, where $k \in \mathbb{N}$. $\mbox{\textsc{Tester}}(G, k)$ takes a graph $G$ as input over a stream along with the parameter $k$ and distinguishes between $\mbox{OPT} \geq k$ and $\mbox{OPT} \leq  k/2$, where $\opt$ denotes the size of the maximum matching in $G$. The meta algorithm runs $\Oh(\log n)$ instances of $\mbox{\textsc{Tester}}$ in parallel and the final output is obtained from the outputs of the $\Oh(\log n)$ instances of \textsc{Tester}. Each instance of $\mbox{\textsc{Tester}}$ is of the form $\mbox{\textsc{Tester}}(G^p,k_p)$, where $G^p$ denotes the induced subgraph of $G$ by the vertex set chosen with probability $p$ using a hash function chosen uniformly from a four-wise independent hash family and parameter $k_p$ depends on the sampling probability $p$. The maximum value of $k_p$ in any instance of $\mbox{\textsc{Tester}}$ is $\tOh(n/\alpha^2)$. Note that the space complexity of $\mbox{\textsc{Tester}}({G, k})$ (to be discussed below) is $\tOh(k)$ and $\tOh(k^2)$ in insertion-only and dynamic streams, respectively. This implies that the space complexity of the streaming algorithm is $\tOh(n/\alpha^2)$ and $\tOh(n^2/\alpha^4)$ in insertion-only and dynamic streams, respectively.

\paragraph{$\mbox{\textsc{Tester}} (G, k)$ in insertion-only streams.}
Here the algorithm either maintains a maximal matching or a matching of size $k$ in $G$. This can be achieved by adding edges to the current matching if the size of the current matching is less than $k$. From the size of the matching stored, the algorithm reports the output.


\paragraph{$\mbox{\textsc{Tester}} (G, k)$ in dynamic streams.}
Here the algorithm chooses a hash function $h: V \rightarrow [\Theta(k)]$ from a pairwise independent hash family uniformly at random. Note that each $h$ partitions the vertex sets into $\Theta(k)$ groups $V_i=\{v \in V | h(v)=i\}$. For each pair $V_i$ and $V_j$, the algorithm maintains an $\ell_0$-sampler for the edges between each $V_i$ and $V_j$, i.e., a sketch for the edge set $E \left(V_i, V_j\right)$ of size $\Oh(\log ^3 n)$ as stated in \Cref{lem:sketch-A-B}. Let $H$ be the subgraph obtained from the outcomes of the $\Theta(k^2)$ many sketches. From the maximum matching of this graph $H$, the algorithm reports the output.

Now we will prove \Cref{thm:match-size}.

\begin{proof}[Proof of \Cref{thm:match-size}]
As \cite{AKL21}, our algorithm in the \MPC model runs $\Oh(\log n)$ instances of \textsc{Tester}. So, we will be done by explaining how we implement $\mbox{\textsc{Tester}}(G,k)$ in the insertion-only \MPC model \hide{such that the global space is same as the space complexity of $\mbox{\textsc{Tester}}(G,k)$ in the insertion-only model, i.e., $\tOh(k)$}. The implementation is similar to the proof of \Cref{thm:match-find}. Let $M$  be the matching stored by the algorithm.  On each machine, we have a set of edges and the information about the edges that are present in the current matching of size at most $k$. Now, let us discuss how to update the information over the machines when a batch $I$ of $\Oh(n^{\spacexp})$ insertions arrive. If $|M|= k$, we do not update anything. Otherwise, we proceed as follows. We broadcast $I$ to all machines and listen from the machines about the edges (in $I$) whose vertices are in some edges in $M$. Note that this can be performed in $\Oh(1)$ rounds. Let $I' \subseteq I$ be the set of edges none of whose endpoints are present in $M$. We greedily add the edges in $I'$ to the current maximal matching till the size of the matching reaches $k$. When we are asked to report a matching, we output the apprximation to the size of the maximum matching of the input graph from the size of the currently stored matching in the memory. Observe that the algorithm uses $\tOh(k)$ space. The correctness and the round complexity of the algorithm follow from the description.
\end{proof}

Apart from the streaming algorithm by \cite{AKL21}, our algorithm of \Cref{thm:match-size1} relies on the algorithm in the \MPC model \cite{NO21}, as stated in \Cref{prop:onak}. The high-level idea is to maintain the information same as that of \cite{AKL21} and run the \MPC algorithm by \cite{NO21} on the graph $H$ formed by the edges we get from the $\Theta(k^2)$ sketches, between the vertex partitions $V_1, \ldots,V_{\Theta(k)}$, in the description of $\mbox{\textsc{Tester}}(G,k)$ in dynamic stream.

\begin{proof}[Proof of \Cref{thm:match-size1}]
Similar to the proof of \Cref{thm:match-size}, here also we will be done by explaining how we implement $\mbox{\textsc{Tester}}(G,k)$ in the  \MPC model when we allow arbitrary updates. Note that the hash function $h: V \rightarrow [\Theta(k)]$ can be stored in every machine since $\Oh(\log n)$ bits are enough to store it. So, the partition of the vertex set into $V_1,\ldots, V_{\Theta(k)}$ (due to $h$) is implicitly stored in every machine. The sketches/$\ell_0$-samplers (each of size $\Oh(\log ^3 n)$) corresponding to $(V_i, V_j)$ pairs are also stored over the machines in a distributed fashion. Our algorithm also stores the outcome of each $\ell_0$-sampler at any instance of time along with the sketches.

Let $U$ be the set of updates under consideration. Let $X$ be the outcome of the $\ell_0$-samplers corresponding to the pairs $V_{h(u)}$ and $V_{h(v))}$ where $\{u,v\} \in U$ (to be either inserted or deleted). Note that $\size{X}\leq \size{U}=\Oh(\lspace^{1-\kappa})$. Now, we gather $X$ in $\Oh(1)$ rounds and delete $X$ from the graph $H$ in $\Oh(\log 1/\kappa)$ rounds using the algorithm of \Cref{prop:onak}. For each edge $e=\{u,v\} \in U$ (either $e$ is to be inserted or deleted), we update the sketch for $E \left(V_{h(u)}, V_{h(v))}\right)$ accordingly. Observe that this can be done by just broadcasting $U$ to every machine in $\Oh(1)$ rounds. Let $Y$ be the outcome of the $\ell_0$-samplers corresponding to pairs $V_{h(u)}$ and $V_{h(v)}$ where $\{u,v\} \in U$. Note that $\size{Y}\leq \size{U}=\Oh(\lspace^{1-\kappa})$. Now, we gather $Y$ in $\Oh(1)$ rounds and insert $Y$ to the graph $H$ in $\Oh(\log 1/\kappa)$ rounds by using the algorithm of \Cref{prop:onak}. Again, we point out that the input to the \MPC algorithm of \Cref{prop:onak} is dependent on the randomness of the sketches/$\ell_0$-samplers. However,  this is not a problem as the random bits of the \MPC algorithm and that of the sketches are independent. When we get a query to report an $\Oh(\alpha)$-approximate matching, we use the output of the \MPC algorithm of \Cref{prop:onak} to report a $2$-approximate matching in $H$, from which we can report an $\Oh(\alpha)$ approximation to the size of the maximum matching. This is possible from the fact that \cite{AKL17} showed that from the size of the maximum matching of $H$ one can report $\Oh(\alpha)$-approximation to the maximum matching of the input graph.

From the above description, the number of rounds to update a batch of size $\Oh(\lspace^{1-\kappa})$ and to report an $\Oh(\alpha)$-approximate matching is $\Oh(\log 1/\kappa)$. As \cite{AKL21} generates a space graph $H$ with size $\tOh(k^2)$ to implement $\mbox{\textsc{Tester}}(G,k)$ and we are using \MPC algorithm of \Cref{prop:onak} on graph $H$ (that uses total memory of $\tOh(|E(H)|)$), the total memory used by our algorithm is $\tOh(k^2)$.
 to simulate $\mbox{\textsc{Tester}}(G,k)$. As we run $\Oh(\log n)$ different instances of $\mbox{\textsc{Tester}}(G,k)$ in parallel with maximum possible $k=\tOh(n/\alpha)$, the total memory used by our algorithm is $\tOh(n^2/\alpha^4)$.
%
\end{proof}



\section{Conclusion}
\label{sec:conclusion}

In this paper, we introduce the Massively Parallel Computation (\MPC) model for data streams. We show that one can efficiently process large batches of edge insertions and deletions for several fundamental graph problems (connectivity, minimum spanning forest, and approximate matching) in a constant number of \MPC rounds using very little memory; the total memory used in our~algorithms matches (up to polylog~factors) the space complexity of the best-known streaming counterparts.


Nevertheless, it is important to note that not all streaming algorithms can be effortlessly adapted to our framework. Examples include problems such as $k$-vertex connectivity, $k$-edge connectivity, minimum cut, and others, which already have efficient semi-streaming algorithms. Exploring the possibility of extending our connectivity result to address these problems would undoubtedly be interesting. It would be equally fascinating to explore advancements in the lower-bound aspect within this model. It is worth mentioning that the streaming lower bounds we have discussed apply to our scenario regardless of the number of rounds spent for updates. {An especially intriguing question is to identify a problem that has a semi-streaming algorithm, yet requires in our \MPC streaming model either a non-constant number of rounds for the updates or a substantial memory requirement. Notice that while proving an unconditional lower bound seems to be very hard since it would imply strong separations in circuit complexity \cite{RVW18}, proving any non-trivial conditional lower bounds would be very interesting.}

While the main focus of our study is on the \emph{one-pass streaming \MPC algorithms}, one could also consider a scenario with \emph{multiple-passes}, in a similar way as one studies graph streaming algorithms. We leave this avenue as an interesting line of future research. 

\paragraph{Acknowledgements}
The authors would like to thank Robi Krauthgamer and Graham Cormode for their useful comments about the model and about streaming algorithms. We also thank the anonymous reviewers for their careful reading of our manuscript and their many insightful comments and suggestions.

\bibliographystyle{alpha}
\bibliography{references}
\end{document}